\documentclass[journal,onecolumn,12pt]{IEEEtran}

\usepackage{graphicx}
\usepackage{times}
\usepackage{amsmath}
\usepackage{setspace}
\usepackage{subfigure}
\usepackage{longtable}
\usepackage{multicol}
\usepackage{amssymb}
\usepackage{epsfig,latexsym,subfigure,cite}
\usepackage{color}

\newtheorem{theorem}{Theorem}
\newtheorem{corollary}{Corollary}
\newtheorem{lemma}{Lemma}

\newtheorem{remark}{Remark}
\newtheorem{example}{Example}

\providecommand{\E}{{\rm E}} \providecommand{\Var}{{\rm Var}}
\providecommand{\Cov}{{\rm Cov}}

\providecommand{\xv}{\mathbf{x}} 
 \providecommand{\vv}{\mathbf{v}}
 \providecommand{\Uv}{\mathbf{U}}
\providecommand{\Vv}{\mathbf{V}} \providecommand{\Xv}{\mathbf{X}}

\providecommand{\hv}{\mathbf{h}} \providecommand{\gv}{\mathbf{g}}
 \providecommand{\bv}{\mathbf{b}}
\providecommand{\ev}{\mathbf{e}} \providecommand{\cv}{\mathbf{c}}

 \providecommand{\Yc}{{\mathcal Y}}
\providecommand{\Wc}{{\mathcal W}} 
 
 \providecommand{\Cc}{{\mathcal C}}
\providecommand{\Pc}{{\mathcal P}} \providecommand{\Rc}{{\mathcal R}}
 \providecommand{\Sc}{{\mathcal S}}
 
\providecommand{\Lc}{{\mathcal L}}

\providecommand{\Yt}{\tilde{Y}} \providecommand{\yt}{\tilde{y}}
\providecommand{\Zt}{\tilde{Z}} \providecommand{\zt}{\tilde{z}}

\providecommand{\Pr}{{\rm Pr}}

\providecommand{\tr}{{\rm tr}}

\title{\LARGE{\rm Secrecy Capacity Region of a Multi-Antenna Gaussian\\[2mm] Broadcast Channel with Confidential Messages}}
\author{Ruoheng Liu and H. Vincent Poor%
\thanks{This research was supported by the National Science Foundation under Grants ANI-03-38807 and CNS-06-25637.
The material in this paper was presented in part at the First International Workshop on Information Theory
for Sensor Networks, Santa Fe, NM, June 18 - 20, 2007}%
\thanks{Ruoheng Liu and H. Vincent Poor are with Department of Electrical Engineering, Princeton University,
Princeton, NJ 08544, USA, email: {\{rliu,poor\}@princeton.edu}.}}

\begin{document}

\maketitle
\thispagestyle{empty}
%\pagestyle{empty}

%%%%%%%%%%%%%%%%%%%%
\doublespace

\begin{abstract}

In wireless data networks, communication is particularly susceptible to
eavesdropping due to its broadcast nature. Security and privacy systems have
become critical for wireless providers and enterprise networks. This paper
considers the problem of secret communication over the Gaussian broadcast
channel, where a multi-antenna transmitter sends independent confidential
messages to two users with \emph{information-theoretic secrecy}. That is, each
user would like to obtain its own confidential message in a reliable and safe
manner. This communication model is referred to as the multi-antenna Gaussian
broadcast channel with confidential messages (MGBC-CM). Under this
communication scenario, a secret dirty-paper coding scheme and the
corresponding achievable secrecy rate region are first developed based on
Gaussian codebooks. Next, a computable Sato-type outer bound on the secrecy
capacity region is provided for the MGBC-CM. Furthermore, the Sato-type outer
bound prove to be consistent with the boundary of the secret dirty-paper coding
achievable rate region, and hence, the secrecy capacity region of the MGBC-CM
is established. Finally, two numerical examples demonstrate that both users can
achieve positive rates simultaneously under the information-theoretic secrecy
requirement.
\end{abstract}


\section{Introduction}

The need for efficient, reliable, and secret data communication over wireless
networks has been rising rapidly for decades. Due to its broadcast nature,
wireless communication is particularly susceptible to eavesdropping. The
inherent problematic nature of wireless networks exposes not only the risks and
vulnerabilities that a malicious user can exploit and severely compromise the
network, but also multiplies information confidentiality concerns with respect
to in-network terminals. Hence, security and privacy systems have become
critical for wireless providers and enterprise networks.

In this work, we consider multiple antenna secret broadcast in wireless
networks. This research is inspired by the seminal paper \cite{Wyner:BSTJ:75},
in which Wyner introduced the so-called {\it wiretap channel} and proposed an
information theoretic approach to secret communication schemes. Under the
assumption that the channel to the eavesdropper is a degraded version of that
to the desired receiver, Wyner characterized the capacity-secrecy tradeoff for
the discrete memoryless wiretap channel and showed that secret communication is
possible without sharing a secret key. Later, the result was extended by
Csisz{\'{a}r and K{\"{o}rner who determined the secrecy capacity for the
non-degraded {\it broadcast channel} (BC) with a single confidential message
intended for one of the users \cite{Csiszar:IT:78}.

In more general wireless network scenarios, secret communication may involve
multiple users and multiple antennas. Motivated by wireless communication,
where transmitted signals are broadcast and can be received by all users within
the communication range, a significant research effort has been invested in the
study of the information-theoretic limits of secret communication in different
wireless network environments including multi-user communication with
confidential messages
\cite{Oohama:ITW:01,Csiszar:IT:04,Tekin:ISIT:06,Tekin:ITA:07,Liang06it,Liu:ISIT:06,Liu:Allerton:06,Lai:IT:06,Elza:CISS:07},
secret wireless communication on fading channels
\cite{Barros:ISIT:06,Li:Allerton:06,Liang06novit,Gopala:ISIT:07}, and the
Gaussian multiple-input single-output (MISO) and multiple-input multiple-output
(MIMO) wiretap channels
\cite{Li:CISS:07,Liu:WITS:07,Wornell:ISIT:07,Wornell-IT-2007,Ulukus:ISIT:07,shafiee-IT-2007}.

These issues motivate us to study the multi-antenna Gaussian BC with
confidential messages (MGBC-CM), in which independent confidential messages
from a multi-antenna transmitter are to be communicated to two users. The
corresponding broadcast communication model is shown in Fig.~\ref{fig:gbc}.
Each user would like to obtain its own message reliably and confidentially.
\begin{figure}[t]
 \centerline{\includegraphics[width=0.8\linewidth,draft=false]{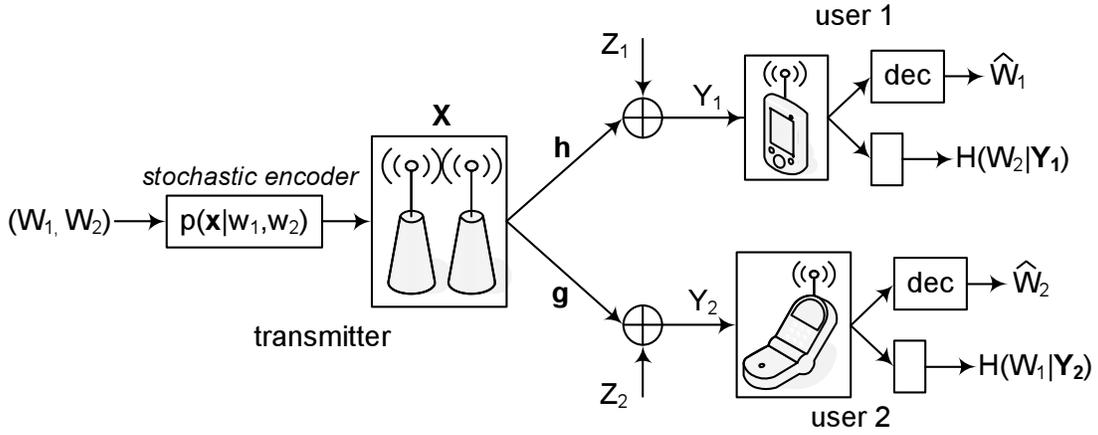}} \caption{
Multiple-antenna Gaussian broadcast channel with confidential message}
  \label{fig:gbc}
\end{figure}

To give insight into this problem, we first consider a single-antenna Gaussian
BC. Note that this channel is degraded~\cite{Cover}, which means that if a
message can be successfully decoded by the inferior user, then the superior
user is also ensured of decoding it. Hence, the secrecy rate of the inferior
user is zero and this problem is reduced to the scalar Gaussian wiretap channel
problem~\cite{Cheong:IT:78} whose secrecy capacity is now the maximum rate
achievable by the superior user. This analysis gives rise to the question: can
the transmitter, in fact, communicate with both users confidentially at nonzero
rate under some other conditions? Roughly speaking, the answer is in the
affirmative. In particular, the transmitter can communicate when equipped with
sufficiently separated multiple antennas.

We here have two goals motivated directly by questions arising in practice. The
first is to determine the condition under which both users can obtain their own
confidential messages %reliably and confidentially
in a reliable and safe manner. This is equivalent to evaluating the secrecy
capacity region for the MGBC-CM. The second is to show \emph{how} the
transmitter should broadcast confidentially, which is equivalent to designing
an achievable secret coding scheme. To this end, we first describe a secret
{\it dirty-paper coding} (DPC) scheme and derive the corresponding achievable
secrecy rate region based on Gaussian codebooks. The secret DPC is based on
{\it double-binning} \cite{Liu07it} which enables both joint encoding and
preserving confidentiality. Next, a computable Sato-type outer bound on the
secrecy capacity region is developed for the MGBC-CM. Furthermore, the
Sato-type outer bound prove to be consistent with the boundary of the secret
dirty-paper coding achievable rate region, and hence, the secrecy capacity
region of the MGBC-CM is established. Finally, two numerical examples
demonstrate that both users can achieve positive rates simultaneously under the
information-theoretic secrecy requirement.

The remainder of this paper is organized as follows. The system model and
definitions are introduced in Section~\ref{sec:model}. The main results on the
secrecy capacity region of the MGBC-CM is state in Section~\ref{sec:main}. The
achievability proof associated with the secret DPC scheme is established in
Section~\ref{sec:in}. The converse proof is derived in Section~\ref{sec:out}
based on the Sato-type outer bound. Finally, Section~\ref{sec:ex} shows
numerical examples and Section~\ref{sec:con} points our our conclusions.

\section{System Model and Definitions} \label{sec:model}

\subsection{Channel Model}

We consider the communication of confidential messages to two users over a
Gaussian BC via $t\ge 2$ transmit-antennas. Each user is equipped with a single
receive-antenna. As shown in Fig.~\ref{fig:gbc}, the transmitter sends
independent confidential messages $W_1$ and $W_2$ in $n$ channel uses with
$nR_1$ and $nR_2$ bits, respectively. The message $W_1$ is destined for user 1
and eavesdropped by user 2, whereas the message $W_2$ is destined for user 2
and eavesdropped by user 1. This communication scenario is referred to as the
{\it multi-antenna Gaussian BC with confidential messages}. The Gaussian BC is
an additive noise channel and the received symbols at user 1 and user 2 are
represented using the following expression:
\begin{align}
y_{1,i}&=\hv^{H} \xv_i+z_{1,i} \notag\\
y_{2,i}&=\gv^{H} \xv_i+z_{2,i}, \qquad  i=1,\dots,n \label{eq:miso}
\end{align}
where $\xv_i \in \mathbb{C}^{t}$ is a complex input vector at time $i$,
$\{z_{1,i}\}$ and $\{z_{2,i}\}$ correspond to two independent, zero-mean,
unit-variance, complex Gaussian noise sequences, and $\hv, \gv\in
\mathbb{C}^{t}$ are fixed, complex channel attenuation vectors imposed on
user~1 and user~2, respectively. The channel input is constrained by
$\tr(K_{\Xv}) \le P$, where $P$ is the average total power limitation at the
transmitter. We also assume that both the transmitter and users are aware of
the attenuation vectors.

\subsection{Important Channel Parameters for the MGBC-CM}

For the MGBC-CM, we are interested in the following important parameters, which
are related to the generalized eigenvalue problem (see Appendix~\ref{app:geig}
for the details).

Let $\lambda_1$ and $\ev_1$ denote the largest generalized eigenvalue and the
corresponding normalized eigenvector of the pencil $(I+P\hv \hv^{H}, I+P\gv
\gv^{H})$ so that $\ev_1^{H}\ev_1=1$ and
\begin{align}
(I+P\hv \hv^{H})\ev_1= \lambda_1(I+P\gv \gv^{H})\ev_1. \label{eq:eig-def1}
\end{align}
Similarly, we define $\lambda_2$ and $\ev_2$ as the largest generalized
eigenvalue and the corresponding normalized eigenvector of the pencil $(I+P\gv
\gv^{H}, I+P\hv \hv^{H})$ so that $\ev_2^{H}\ev_2=1$ and
\begin{align}
(I+P\gv \gv^{H})\ev_2= \lambda_2(I+P\hv \hv^{H})\ev_2. \label{eq:eig-def2}
\end{align}
An useful property of $\lambda_1$ and $\lambda_2$ is described in the following
lemma.
\begin{lemma} \label{lem:geig}
For any channel attenuation vector pair $\hv$ and $\gv$, the largest
generalized eigenvalues of the pencil $(I+P\hv \hv^{H}, I+P\gv \gv^{H})$ and
the pencil $(I+P\gv \gv^{H}, I+P\hv \hv^{H})$ satisfy
\begin{align}
\lambda_1 \ge 1 \quad \text{and} \quad \lambda_2 \ge 1.
\end{align}
Moreover, if $\hv$ and $\gv$ are linearly independent, then both $ \lambda_1$
and $\lambda_2$ are strictly greater than $1$.
\end{lemma}
\begin{proof}
We provide the proof in Appendix~\ref{app:geig}.
\end{proof}

\subsection{Definitions}

We now define the secret codebook, the probability of error, the secrecy level,
and the secrecy capacity region for the MGBC-CM as follows.

An $(2^{nR_1},2^{nR_2},n)$ {\it secret codebook} for the MGBC-CM consists of
the following:
\begin{enumerate}
  \item Two message sets $\Wc_1=\{1,\ldots,2^{nR_1}\}$ and
  $\Wc_2=\{1,\ldots,2^{nR_2}\}$.

  \item An stochastic encoding function is specified by a
  matrix of conditional probability density $p(\xv^{n}|w_1,w_2)$,
  where $\xv^{n}=[\xv_1,\ldots,\xv_n]\in \mathbb{C}^{t\times n}$, $w_k\in\Wc_k$, and
\begin{align*}
\int_{\xv^{n}}p(\xv^{n}|w_1,w_2)=1.
\end{align*}

  \item Decoding functions $\phi_1$ and $\phi_2$. The decoding function
  at user $k$ is a deterministic mapping $$\phi_k  : \Yc_k^{n}\rightarrow
  \Wc_k.$$
\end{enumerate}

\begin{remark}
To increase the randomness of transmitted messages, we consider a {\it
stochastic} encoder at the transmitter. In other words, $p(\xv^{n}|w_1,w_2)$ is
the conditional probability density that the messages $(w_1,w_2)$ are jointly
encoded as the channel input sequence $\xv^{n}$.
\end{remark}

At the receiver ends, the error performance and the secrecy level are evaluated
by the following performance measures.
\begin{enumerate}
  \item The reliability is measured by the maximum error probability
  $$P_e^{(n)}\triangleq \max \bigl\{P_{e,1}^{(n)}, P_{e,2}^{(n)}\bigr\}$$
  where $P_{e,k}^{(n)}$ is the error probability for user $k$ given by
\begin{align}
P_{e,k}^{(n)}&=2^{-n(R_1+R_2)}\sum_{w_1\in\Wc_1}\sum_{w_2\in\Wc_2}\Pr\bigl[\phi_k(Y_k^{n})\neq
w_k \big|(w_1,w_2)  \text{ sent} \bigr].
\end{align}

  \item The secrecy levels with respect to confidential messages $W_1$ and $W_2$
are measured, respectively, at user 2 and user 1 with respect to the {\it
equivocation rates}
\begin{equation} \label{eq:nq}
\frac{1}{n}H(W_2|Y_1^{n})\quad \text{and} \quad \frac{1}{n}H(W_1|Y_2^{n}).
\end{equation}
\end{enumerate}

A rate pair $(R_1,  R_2)$ is said to be achievable for the MGBC-CM if, for any
$\epsilon>0$, there exists an $(2^{nR_1}, 2^{nR_2}, n)$ code that satisfies
$P_e^{(n)}\le \epsilon$, and the information-theoretic secrecy
requirement\footnote{This definition corresponds to the so-called {\it weak
secrecy-key rate} \cite{Maurer:EUROCRYPT:00}. A stronger measurement of the
secrecy level has been defined by Maurer and Wolf in terms of absolute
equivocation \cite{Maurer:EUROCRYPT:00}, where the authors have shown that the
former definition could replaced by the latter without any rate penalty in a
wiretap channel.}
\begin{align}
&            & nR_1-H(W_1|Y_2^{n})&\le n\epsilon &\notag\\
&\text{and}  & nR_2-H(W_2|Y_1^{n})&\le n\epsilon. & \label{eq:equiv}
\end{align}

The {\it secrecy capacity region} $\Cc^{\rm MG}_{s}$ of the MGBC-CM is the
closure of the set of all achievable rate pairs $(R_1, R_2)$.

\section{Main Result: Secrecy Capacity Region for the MGBC-CM} \label{sec:main}

The two-user Gaussian BC with multiple transmit-antennas is non-degraded. For
this channel, we have the following closed-from result on the secrecy capacity
region under the information-theoretic secrecy requirement.

\begin{theorem} \label{thm:GBC}
We consider an MGBC-CM modeled in (\ref{eq:miso}). Let
\begin{align}
\gamma_1(\alpha)=\frac{1+\alpha P|\hv^{H}\ev_1|^2}{1+ \alpha P
|\gv^{H}\ev_1|^2}, \label{eq:def-gm-1}
\end{align}
$\gamma_2(\alpha)$ be the largest generalized eigenvalue of the pencil
\begin{align}
\left(I+\frac{(1-\alpha) P}{1+\alpha P |\gv^{H}\ev_1|^2}\gv\gv^{H}, \;
I+\frac{(1-\alpha) P}{1+ \alpha P |\hv^{H}\ev_1|^2}\hv \hv^{H}\right),
\label{eq:pencil2}
\end{align}
and ${\Rc}^{\rm MG}(\alpha)$ denote the union of all $(R_1,R_2)$ satisfying
\begin{align}
&          &  0&\le R_1\le \log_2 \gamma_1(\alpha)  &\notag\\
&\text{and} & 0&\le R_2\le \log_2 \gamma_2(\alpha).  & \label{eq:sc-rate}
\end{align}
The secrecy capacity region of the MGBC-CM is
\begin{align}
\Cc^{\rm MG}_{s}= {\rm co} \left\{\bigcup_{0\le \alpha \le 1} \Rc^{\rm
MG}(\alpha)\right\}, \label{eq:sc-mg}
\end{align}
where ${\rm co}\{\Sc\}$ denotes the convex hull of the set $\Sc$.
%where $P$ is the average total power and $\alpha$ is the power allocation parameter.
\end{theorem}
\begin{proof}
We provide the achievability proof in Section~\ref{sec:in} based on a secret
dirty paper coding scheme, and show the converse proof in Section~\ref{sec:out}
based on a Sato-type outer bound.
\end{proof}

Based on Theorem~\ref{thm:GBC}, we can calculate the boundary of the secrecy
capacity region $\Cc^{\rm MG}_{s}$ by choosing $\alpha$ to trade off the rate
$R_1$ for the rate $R_2$. In particular, when $\alpha=1$, we obtain
\begin{align}
&           & \gamma_1(1)&=\frac{\ev_1^{H}(I+P\hv\hv^{H})\ev_1}{\ev_1^{H}(I+P\gv\gv^{H})\ev_1}=\lambda_1&\label{eq:cal-g1}\\
& \text{and}& \gamma_2(1)&=1&
\end{align}
where (\ref{eq:cal-g1}) follows from the definitions of $\lambda_1$ and $\ev_1$
in (\ref{eq:eig-def1}). Theorem~\ref{thm:GBC} implies that the rate pair
$(\log_2\lambda_1, 0)$ is achievable. In fact, this rate pair is the corner
point corresponding to the maximum achievable rate of user 1 in the capacity
region $\Cc^{\rm MG}_{s}$.
\begin{corollary} \label{cor:maxr1}
For the MGBC-CM, the maximum secrecy rate of user 1 is given by
\begin{align}
R_{1,\max} =\max_{0\le\alpha\le1} \log_2\gamma_1(\alpha) = \log_2\lambda_1
\end{align}
where $\lambda_1$ is the largest generalized eigenvalue of the pencil $(I+P\hv
\hv^{H}, I+P\gv \gv^{H})$.
\end{corollary}
\begin{proof}
See Appendix~\ref{app:sec3}.
\end{proof}
%Corollary~\ref{cor:maxr1} illustrates that user 1 achieves the maximum secrecy
%rate when user 2 is not obtained any confidential information.

\begin{example} (MISO Wiretap Channels)
A special case of the MGBC-CM model is the Gaussian MISO wiretap channel
studied in \cite{Li:CISS:07,Wornell:ISIT:07,Ulukus:ISIT:07}, where the
transmitter sends confidential information to only one user and treats another
user as an eavesdropper. Let us consider a Gaussian MISO wiretap channel
modeled in (\ref{eq:miso}), where user 1 is the legitimate receiver and user 2
is the eavesdropper. Corollary~\ref{cor:maxr1} implies that the secrecy
capacity of the Gaussian MISO wiretap channel corresponds to the corner point
of $\Cc^{\rm MG}_{s}$. Hence, the secrecy capacity of the Gaussian MISO wiretap
channel is given by
\begin{align}
C^{\rm MISO}_{s}=\log_2 \lambda_1,
\end{align}
which coincides with the result of \cite{Wornell:ISIT:07}.
\end{example}

For the MGBC-CM, the actions of user 1 and user 2 are symmetric to each other,
i.e., each user decodes its own message and eavesdrops the confidential
information belonging to another user. Based on symmetry of this two-user BC
model, we can express the secrecy capacity region $\Cc^{\rm MG}_{s}$ in an
alternative way.
\begin{corollary} \label{cor:GBC-2}
For an MGBC-CM modeled in (\ref{eq:miso}), the secrecy capacity region can be
written as
\begin{align}
\Cc^{\rm MG}_{s}= {\rm co} \left\{\bigcup_{0\le \beta \le 1} \Rc^{\rm
MG-2}(\beta)\right\} \label{eq:sc-mg-2}
\end{align}
where ${\Rc}^{\rm MG-2}(\beta)$ denotes the union of all $(R_1,R_2)$ satisfying
\begin{align}
&          &  0&\le R_1\le \log_2 \xi_1(\beta) &\notag\\
&\text{and} & 0&\le R_2\le \log_2 \xi_2(\beta), & \label{eq:sc-rate-2}
\end{align}
$\xi_1(\beta)$ is the largest generalized eigenvalue of the pencil
\begin{align}
\left(I+\frac{(1-\beta) P}{1+\beta P |\hv^{H}\ev_2|^2}\hv\hv^{H},\;
I+\frac{(1-\beta) P}{1+ \beta P |\gv^{H}\ev_2|^2}\gv \gv^{H}\right)
\label{eq:pencil3}
\end{align}
and
\begin{align}
\xi_2(\beta)=\frac{1+\beta P|\gv^{H}\ev_2|^2}{1+ \beta P |\hv^{H}\ev_2|^2}.
\end{align}
\end{corollary}
\begin{proof}
The derivation follows from the same approach of the proof for
Theorem~\ref{thm:GBC} by reversing the roles of user $1$ and user $2$.
\end{proof}

\begin{remark}
Theorem~\ref{thm:GBC} and Corollary~\ref{cor:GBC-2} imply that if $\alpha$ and
$\beta$ satisfy the implicit function $\gamma_1(\alpha)=\xi_1(\beta)$, then
$${\Rc}^{\rm MG}(\alpha)={\Rc}^{\rm MG-2}(\beta).$$
For example, it is easy to check ${\Rc}^{\rm MG}(1)={\Rc}^{\rm MG-2}(0)$.
\end{remark}

Now, by applying Corollary~\ref{cor:GBC-2} and setting $\beta=1$, we can show
that the rate pair $(0,  \log_2\lambda_2)$ is the corner point corresponding to
the maximum achievable rate of user 2 in the capacity region $\Cc^{\rm
MG}_{s}$.
\begin{corollary} \label{cor:maxr2}
For the MGBC-CM, the maximum secrecy rate of user 2 is given by
\begin{align}
R_{2,\max} =\log_2\lambda_2
\end{align}
where $\lambda_2$ is the largest generalized eigenvalue of the pencil $(I+P\gv
\gv^{H}, I+P\hv \hv^{H})$.
\end{corollary}
\begin{proof}
The derivation follows from the same approach of the proof for
Corollary~\ref{cor:maxr1}.
\end{proof}
Corollaries~\ref{cor:maxr1} and~\ref{cor:maxr2} imply that for the MGBC-CM,
both users can achieve positive rates with information-theoretic secrecy if and
only if $\lambda_1>1$ and $\lambda_2>1$. Lemma~\ref{lem:geig} illusrtates that
this condition can be ensured when the attenuation vectors $\hv$ and $\gv$ are
linear independent.

\section{Secret DPC Coding Scheme and Achievability Proof} \label{sec:in}

We first briefly review the prior information-theoretic result on the
achievable rate region for the {\it BC with confidential messages} (BC-CM) of
\cite{Liu07it}. Based on this result, we develop the achievable secret coding
scheme for the MGBC-CM and find the capacity achieving input covariance matrix.

\subsection{Double-Binning Inner bound for the BC-CM}

An achievable rate region for the BC-CM has been established in \cite{Liu07it}
based on a double-binning scheme that enables both joint encoding at the
transmitter by using Slepian-Wolf binning \cite{Slepian:IT:73} and preserving
confidentiality by using random binning. We summarize the double-binning
codebook and encoding strategy in Appendix~\ref{app:sec4} for completeness.

\begin{lemma} (\cite[Theorem~3]{Liu07it}) \label{lem:InBC}
Let $\Vv_1$ and $\Vv_2$ be auxiliary random variables, $\Omega$ denote the
class of joint probability densities $p(\vv_1,\vv_2,\xv,y_1,y_2)$ that factor
as
\begin{align}
p(\vv_1,\vv_2)p(\xv|\vv_1,\vv_2)p(y_1,y_2|\xv),
\end{align}
and ${\Rc}_{\rm I}(\pi)$ denote the union of all $(R_1,R_2)$ satisfying
\begin{align}
&            & 0 &\le R_1 \le I(\Vv_1;Y_1)-I(\Vv_1;Y_2|\Vv_2)-I(\Vv_1;\Vv_2) &\label{eq:BC-IN-R1}\\
& \text{and} & 0 &\le R_2 \le I(\Vv_2;Y_2)-I(\Vv_2;Y_1|\Vv_1)-I(\Vv_1;\Vv_2)
\label{eq:BC-IN-R2}
\end{align}
for a given joint probability density $\pi\in \Omega$. For the BC-CM, any rate
pair
\begin{align}
(R_1,R_2)\in {\rm co} \left\{\bigcup_{\pi\in\Omega} \Rc_{\rm I}(\pi)\right\}
\label{eq:inner}
\end{align}
is achievable.
\end{lemma}

The proof of Lemma~\ref{lem:InBC} can be found in \cite{Liu07it}. Here, we
provide an alternative view on this result. Since randomization can increase
secrecy, we employ stochastic encoding at the transmitter so that the size of
the secret codebook is larger than the size of message set. Let $R'$ denote the
redundant rate used to prevent the confidentiality. The best known achievable
region for a general BC was found by Marton of \cite{marton:IT:77}. Now, for a
given joint density $p(\vv_1,\vv_2,\xv)$, a special case of the Marton sum rate
(without a common rate) is given by
\begin{align}
R_1+R_2+R' \le I(\Vv_1;Y_1)+I(\Vv_2;Y_2)-I(\Vv_1;\Vv_2). \label{eq:marton}
\end{align}
On the other hand, the total (both the intended and the eavesdropped)
information rate obtained by user~2 is limited by $I(\Vv_1,\Vv_2;Y_2)$.
Intuitively, to keep the message $W_1$ secret from user~2, the redundant rate
$R'$ should satisfy that
\begin{align}
R_2+R' &\ge I(\Vv_1,\Vv_2;Y_2).
\end{align}
This implies that to satisfy the information-theoretic secrecy requirement, the
achievable secrecy rate of user 1 can be written as
\begin{align}
R_1 &\le [I(\Vv_1;Y_1)+I(\Vv_2;Y_2)-I(\Vv_1;\Vv_2)]-I(\Vv_1,\Vv_2;Y_2).
\label{eq:epin1}
\end{align}
Similarly, the achievable secrecy rate of user 2 can be written as
\begin{align}
R_2 \le [I(\Vv_1;Y_1)+I(\Vv_2;Y_2)-I(\Vv_1;\Vv_2)]-I(\Vv_1,\Vv_2;Y_1).
\label{eq:epin2}
\end{align}
Bounds (\ref{eq:epin1}) and (\ref{eq:epin2}) lead to the achievable secrecy
rate region in Lemma~\ref{lem:InBC}.

\begin{remark}
For the BC with confidential messages, one can employ joint encoding at the
transmitter. However, to preserve confidentiality, both achievable rate
expressions in (\ref{eq:BC-IN-R1}) and (\ref{eq:BC-IN-R2}) include a penalty
term $I(\Vv_1;\Vv_2)$. Hence, compared with Marton's achievable region
\cite{marton:IT:77} for a general BC, here, one need to pay ``double'' for
jointly encoding at the transmitter.
\end{remark}

%The achievable strategy in Lemma~\ref{lem:InBC} introduces a double-binning
%coding scheme for the discrete memoryless BC with confidential message that
%enables both joint encoding at the transmitter by using Slepian-Wolf binning
%\cite{Slepian:IT:73} and preserving confidentiality by using random binning.

\subsection{Secret DPC Scheme for the MGBC-CM}

The achievable strategy in Lemma~\ref{lem:InBC} introduces a double-binning
coding scheme. However, when the rate region (\ref{eq:inner}) is used as a
constructive technique, it not clear how to choose the auxiliary random
variables $\Vv_1$ and $\Vv_2$ to implement the double-binning codebook, and
hence, one has to ``guess'' the density of $p(\vv_1,\vv_2,\xv)$. Here, we
employ the DPC technique with the double-binning code structure to develop the
{\it secret DPC} (S-DPC) achievable rate region for the MGBC-CM.

%achieve secrecy capacity region $\Cc^{\rm MG}_{s}$.

For the MGBC-CM, we consider a secret dirty-paper encoder with Gaussian
codebooks as follows. First, we sperate the channel input $\Xv$ into two random
vectors $\Uv_1$ and $\Uv_2$ so that
\begin{align}
\Uv_1+\Uv_2=\Xv. \label{eq:u1u2}
\end{align}
We choose $\Uv_1$ and $\Uv_2$ as well as auxiliary random variables $\Vv_1$ and
$\Vv_2$ as follows:
\begin{align}
\Uv_1&\thicksim \mathcal{CN}(0,K_{\Uv_1}),\notag\\
\Uv_2&\thicksim \mathcal{CN}(0,K_{\Uv_2}),~\text{independent of }\Uv_1\notag\\
\Vv_1&=\Uv_1+\bv \hv^{H}\Uv_2 \quad \text{and}\quad \Vv_2=\Uv_2\label{eq:rvs}
\end{align}
where $K_{\Uv_1}$ and $K_{\Uv_2}$ are covariance matrices of $\Uv_1$ and
$\Uv_2$, respectively, and
\begin{align}
\bv=\frac{K_{\Uv_1}\hv}{1+\hv^{H} K_{\Uv_1} \hv}.
\end{align}
Based on the conditions (\ref{eq:rvs}) and Lemma~\ref{lem:InBC}, we obtain a
S-DPC rate region for the MGBC-CM as follows.

\begin{lemma}\label{lem:GBCin} {\rm [S-DPC region]}
Let  ${\Rc}_{\rm I}^{\rm S-DPC}(K_{\Uv_1},K_{\Uv_2})$ denote the union of all
$(R_1,R_2)$ satisfying
\begin{align}
& &0&\le R_1\le \log_2 \frac{1+\hv^{H} K_{\Uv_1} \hv}{1+\gv^{H} K_{\Uv_1}
\gv}& \label{eq:dpc-r1}\\
& \text{and} & 0&\le R_2\le \log_2 \frac{1+\gv^{H} (K_{\Uv_1}+K_{\Uv_2})
\gv}{1+\hv^{H} (K_{\Uv_1}+K_{\Uv_2}) \hv}+ \log_2 \frac{1+\hv^{H}
K_{\Uv_1}\hv}{1+\gv^{H} K_{\Uv_1} \gv}. \label{eq:dpc-r2}&
\end{align}
Then, any rate pair
\begin{align}
(R_1,R_2)\in {\rm co} \left\{\bigcup_{\tr(K_{\Uv_1}+K_{\Uv_2})\le P} \Rc_{\rm
I}^{\rm S-DPC}(K_{\Uv_1},K_{\Uv_2})\right\} \label{eq:dpc-r}
\end{align}
is achievable for the MGBC-CM.
\end{lemma}
\begin{proof}
See the Appendix~\ref{app:sec4}.
\end{proof}

\begin{remark}
We choose the random variables $\Uv_1$, $\Uv_2$, $\Vv_1$, $\Vv_2$ and $\Xv$ as
the same as the classical DPC strategy (e.g., see \cite{Caire:IT:03,Yu:IT:04}).
However, the S-DPC scheme is different from the classical one. The codebook and
the coding structure of the S-DPC scheme is based on the double-binning (see
Appendix~\ref{app:sec4}).
\end{remark}

\subsection{Achievability Proof of Theorem~\ref{thm:GBC}}

The S-DPC achievable rate region (\ref{eq:dpc-r}) requires optimization of the
covariance matrices $K_{\Uv_1}$ and $K_{\Uv_2}$. In order to achievable the
boundary of $\Cc^{\rm MG}_{s}$, we choose $K_{\Uv_1}$ and $K_{\Uv_2}$ as
follows:
\begin{align}
&           & K_{\Uv_1}&= \alpha P \ev_1\ev_1^{H} &\notag\\
&\text{and} & K_{\Uv_2}&= (1-\alpha) P \cv_2(\alpha)\cv_2^{H}(\alpha), \qquad
\text{for}~0\le\alpha\le1 & \label{eq:ku12}
\end{align}
where $\ev_1$ is defined in (\ref{eq:eig-def1}) and $\cv_2(\alpha)$ is a
normalized eigenvector of the pencil (\ref{eq:pencil2}) corresponding to
$\gamma_2(\alpha)$ so that $\cv_2^{H}(\alpha)\cv_2(\alpha)=1$ and
\begin{align}
\left(I+\frac{(1-\alpha) P}{1+\alpha P |\gv^{H}\ev_1|^2}\gv\gv^{H}\right)
\cv_2(\alpha) = \gamma_2(\alpha) \left( I+\frac{(1-\alpha) P}{1+ \alpha P
|\hv^{H}\ev_1|^2}\hv \hv^{H}\right) \cv_2(\alpha). \label{eq:g2c2}
\end{align}
Since $\Uv_1$ and $\Uv_2$ are independent, (\ref{eq:u1u2}) implies that the
input covariance matrix can be written as follows:
\begin{align}
K_{\Xv}&= K_{\Uv_1}+K_{\Uv_2}\notag\\
&= \alpha P \ev_1\ev_1^{H}+(1-\alpha) P \cv_2(\alpha)\cv_2^{H}(\alpha), \qquad
\text{for}~0\le\alpha\le1. \label{eq:kx}
\end{align}
Hence, we have
\begin{align}
\tr(K_{\Xv})=\tr(K_{\Uv_1}+K_{\Uv_2})=P,
\end{align}
i.e., the channel input power constraint is satisfied.

Next, inserting (\ref{eq:ku12}) into (\ref{eq:dpc-r1}) and (\ref{eq:dpc-r2}),
we obtain
\begin{align}
\frac{1+\hv^{H} K_{\Uv_1} \hv}{1+\gv^{H} K_{\Uv_1}
\gv}&=\gamma_1(\alpha)\label{eq:sdpc-r1}
\end{align}
and
\begin{align}
\frac{[1+\gv^{H} (K_{\Uv_1}+K_{\Uv_2}) \gv] [1+\hv^{H} K_{\Uv_1}\hv]}
{[1+\hv^{H} (K_{\Uv_1}+K_{\Uv_2}) \hv][1+\gv^{H} K_{\Uv_1} \gv]}
=\gamma_2(\alpha). \label{eq:sdpc-r2}
\end{align}
where the intermediate steps for deriving (\ref{eq:sdpc-r2}) are given in
Appendix~\ref{app:sec4}. Now, by substituting (\ref{eq:sdpc-r1}) and
(\ref{eq:sdpc-r2}) into Lemma~\ref{lem:GBCin}, we obtain the desired achievable
result.

\begin{remark}
The secrecy capacity region $\Cc^{\rm MG}_{s}$ can be achieved by using the
S-DPC scheme, in which the capacity achieving input covariance matrix is with
rank $2$. Furthermore, by reversing the roles of user $1$ and user $2$, we have
the achievability proof for Corollary~\ref{cor:GBC-2}.
\end{remark}

\section{Sato-Type Outer Bound and Converse Proof} \label{sec:out}

In this section, we first describe a new Sato-type outer bound that can be
applied to both discrete memoryless and Gaussian broadcast channels with
confidential messages. Next, a computable Gaussian version of this bound is
derived for the MGBC-CM. Finally, we prove that the Sato-type outer bound
coincides with the secrecy capacity region $\Cc^{\rm MG}_{s}$.

\subsection{Sato-Type Outer Bound}

We consider an important property for the BC-CM in the following lemma.
\begin{lemma} \label{lem:sam}
Let $\Pc$ denote the set of channels $p_{\Yt_1,\Yt_2|\Xv}$ whose marginal
distributions satisfy
\begin{align}
&            & p_{\Yt_1|\Xv}(y_1|\xv)&=p_{Y_1|\Xv}(y_1|\xv) &\notag \\
&\text{and}  & p_{\Yt_2|\Xv}(y_2|\xv)&=p_{Y_1|\Xv}(y_2|\xv) &
\end{align}
for all $y_1$, $y_2$ and $\xv$. The secrecy capacity region $\Cc^{\rm MG}_{s}$
is the same for the channels $p_{\Yt_1,\Yt_2|\Xv} \in \Pc$.
\end{lemma}
\begin{proof}
We provide the proof in Appendix~\ref{app:sec4}.
\end{proof}

We note that $\Pc$ is the set of channels $p_{\Yt_1,\Yt_2|\Xv}$ that have the
same marginal distributions as the original channel transition density
$p_{Y_1,Y_2|\Xv}$. Lemma~\ref{lem:sam} implies that the secrecy capacity region
$\Cc^{\rm MG}_{s}$ depends only on marginal distributions.

\begin{theorem} \label{thm:out1} %{\bf (outer bound)}
Let $\Rc_{\rm O}\bigl(P_{\Yt_1,\Yt_2|\Xv}, P_{\Xv}\bigr)$ denote the union of
all rate pairs $(R_1, R_2)$ satisfying
\begin{align}
&     &        R_1&\le I(\Xv;\Yt_1,\Yt_2)-I(\Xv;\Yt_2) & \label{eq:BC-out-R1}\\
& \text{and} & R_2&\le I(\Xv;\Yt_1,\Yt_2)-I(\Xv;\Yt_1) & \label{eq:BC-out-R2}
\end{align}
for given distributions $P_{\Xv}$ and $P_{\Yt_1,\Yt_2|\Xv}$. The secrecy
capacity region $\Cc^{\rm MG}_{s}$ of the BC-CM satisfies
\begin{align}
\Cc^{\rm MG}_{s} \subseteq \bigcap_{P_{\Yt_1,\Yt_2|\Xv}\in \Pc}
\left\{\bigcup_{P_{\Xv}} \Rc_{\rm O}\bigl(P_{\Yt_1,\Yt_2|\Xv},
P_{\Xv}\bigr)\right\}. \label{eq:Sato}
\end{align}
\end{theorem}
\begin{proof}
See the Appendix~\ref{app:sec4}.
\end{proof}
\begin{remark}
The outer bound (\ref{eq:Sato}) follows by evaluating the secrecy level at each
user end in an individual manner, while by letting the users decode their
messages in a \emph{cooperative} manner. In this sense, we refer to this bound
as ``Sato-type'' outer bound.
\end{remark}

For example, we consider the confidential message $W_1$ that is destined for
user 1 (corresponding to $\Yt_1$) and eavesdropped by user 2 (corresponding to
$\Yt_2$). We assume that a genie gives user 1 the signal $\Yt_2$ as the side
information for decoding $W_1$. Note that the eavesdropped signal $\Yt_2$ at
user 2 is always a degraded version of the entire received signal
$(\Yt_1,\Yt_2)$. This permits the use of the wiretap channel result of
\cite{Wyner:BSTJ:75}.

\begin{remark}
Although Theorem~\ref{thm:out1} is based on a \emph{degraded} argument, the
outer bound (\ref{eq:Sato}) can be applied to \emph{general} broadcast channels
with confidential messages.
\end{remark}

\subsection{Sato-Type Outer Bound for the MGBC-CM}

For the Gaussian BC, the family  $\Pc$ is the set of channels
\begin{align}
\yt_1&=\hv^{H} \xv+\zt_1 \notag\\
\yt_2&=\gv^{H} \xv+\zt_2 \label{eq:miso-2}
\end{align}
where $\zt_1$ and $\zt_2$ correspond to arbitrarily correlated, zero-mean,
unit-variance, complex Gaussian random variables. Let $\rho$ denote the
covariance between $\Zt_1$ and $\Zt_2$, i.e,
$$\Cov\bigl(\Zt_1,\Zt_2\bigr)=\rho \quad \text{and} \quad |\rho|^2\le1.$$
Now, the rate region ${\Rc}_{\rm O}\bigl(P_{\Yt_1,\Yt_2|\Xv}, P_{\Xv}\bigr)$ is
a function of the noise covariance $\rho$ and the input covariance matrix
$K_{\Xv}$. We consider a computable Sato-type outer bound for the MGBC-CM in
the following lemma.
%Based on the fact that Gaussian input distributions maximize ${\Rc}_{\rm
%O}\bigl(P_{\Yt_1,\Yt_2|\Xv}, P_{\Xv}\bigr)$ (see the detail in Appendix~\ref{app:sec5}),
%We describe a computable Sato-type outer bound for the MGBC-CM in the following lemma.
\begin{lemma} \label{lem:outG}
Let ${\Rc}_{\rm O}^{\rm MG}(\rho, K_{\Xv})$ denote the union of all rate pairs
$(R_1,R_2)$ satisfying
\begin{align}
&            & 0 \le R_1&\le f_1(\rho,K_{\Xv}) &\\
& \text{and} & 0 \le R_2&\le f_2(\rho,K_{\Xv})&
\end{align}
where
\begin{align}
& & f_1(\rho,K_{\Xv})& = \min_{\nu \in \mathbb{C}} \log_2
\frac{(\hv-\nu\gv)^{H} K_{\Xv}
(\hv-\nu\gv)+1+|\nu|^2-\nu^{*}\rho-\rho^{*}\nu}{(1-|\rho|^2)}
& \label{eq:f1-def} \\
& \text{and} & f_2(\rho,K_{\Xv})&= \min_{\mu \in \mathbb{C}} \log_2
\frac{(\gv-\mu\hv)^{H} K_{\Xv}
(\gv-\mu\hv)+1+|\mu|^2-\mu^{*}\rho-\rho^{*}\mu}{(1-|\rho|^2)}.
&\label{eq:f2-def}
\end{align}
For the MGBC-CM, the secrecy capacity region $\Cc^{\rm MG}_{s}$ satisfies
\begin{align}
\Cc^{\rm MG}_{s} \subseteq \bigcup_{\tr(K_{\Xv})\le P} {\Rc}_{\rm
O}(\rho,K_{\Xv}) \label{eq:Sato2}
\end{align}
for any $0\le|\rho|\le1$.
\end{lemma}
\begin{proof}
We provide the proof in Appendix~\ref{app:sec4}.
\end{proof}

\subsection{Converse Proof of Theorem~\ref{thm:GBC}}

%Lemma~\ref{lem:outG} describes a computable Sato-type outer bound for the MGBC-CM.

In this subsection, we prove that the Sato-type outer bound of
Lemma~\ref{lem:outG} coincides with the secrecy capacity region $\Cc^{\rm
MG}_{s}$ by properly choosing the parameter $\rho$.
%The proof includes the following steps: we first choose the parameter $\rho$
%used in Lemma~\ref{lem:outG}; next, we establish the relationship between the
%channel input covariance matrix $K_{\Xv}$ and the parameter $\alpha$; finally,
%based on this relationship, we derive that the secrecy capacity region for the
%MGBC-CM.

\subsubsection{Choosing the parameter $\rho$}

Note that Lemma~\ref{lem:outG} is true for any $\rho$ such that $0\le |\rho|\le
1$. In particular, we consider
\begin{align}
\rho_{\rm o}\triangleq \frac{\gv^{H}\ev_1}{\hv^{H}\ev_1}. \label{eq:def-rhoo}
\end{align}
The definitions of $\lambda_1$ and $\ev_1$ in (\ref{eq:eig-def1}) imply that
\begin{align}
|\hv^{H}\ev_1|^2- \lambda_1|\gv^{H}\ev_1|^2&= \frac{\lambda_1-1}{P}.
\label{eq:hhgg}
\end{align}
Since $\lambda_1\ge 1$ (see Lemma~\ref{lem:geig}), we obtain
\begin{align}
\left|\frac{\gv^{H}\ev_1}{\hv^{H}\ev_1}\right| \le 1.
\end{align}
Hence, we can choose $\rho=\rho_{\rm o}$ in Lemma~\ref{lem:outG}.

\subsubsection{Determining the relationship between $K_{\Xv}$ and $\alpha$}

We observe that the rate region ${\Rc}_{\rm O}^{\rm MG}(\rho_{\rm o}, K_{\Xv})$
defined in Lemma~\ref{lem:outG} is a function of the input covariance matrix
$K_{\Xv}$, while the rate region $\Rc^{\rm MG}(\alpha)$ defined in
Theorem~\ref{thm:GBC} is a function of $\alpha$. In order to prove the main
result, we build the relationship between $K_{\Xv}$ and $\alpha$ in the
following lemma.

\begin{lemma} \label{lem:Ka}
For any input covariance matrix $K_{\Xv}$ with $\tr(K_{\Xv})\le P$, there
exists a $\alpha\in[0,1]$ such that $L(K_{\Xv},\alpha)=0$, where
\begin{align}
L(K_{\Xv},\alpha)=[\hv-\rho_{\rm o}\gamma_1(\alpha)\gv]^{H} (K_{\Xv}-\alpha P
\ev_1\ev_1^{H})[\hv-\rho_{\rm o}\gamma_1(\alpha)\gv].
\end{align}
\end{lemma}
\begin{proof}
We provide the proof in Appendix~\ref{app:sec5}.
\end{proof}
Based on the function $L(\cdot)$, we define the subset of input covariance
matrices in terms of $\alpha$ as follows:
\begin{align}
\Lc(\alpha)=\{K_{\Xv}:~L(K_{\Xv},\alpha)=0\}.
\end{align}
Moreover, Lemma~\ref{lem:Ka} implies that
\begin{align}
\bigcup_{0\le\alpha\le 1}\Lc(\alpha)=\{K_{\Xv}:~\tr(K_{\Xv})\le P\}.
\label{eq:lset}
\end{align}

\subsubsection{Bound on $f_1(\rho_{\rm o},K_{\Xv})$}
Now, we prove that if $K_{\Xv}\in \Lc(\alpha)$, then
\begin{align}
f_1(\rho_{\rm o},K_{\Xv}) \le \log_2 \gamma_1(\alpha).
\end{align}

Let $\nu(\alpha)=\rho_{\rm o} \gamma_1(\alpha)$. For a given
$K_{\Xv}\in\Lc(\alpha)$, the definition (\ref{eq:f1-def}) implies that
\begin{align}
f_1(\rho_{\rm o},K_{\Xv}) & \le \log_2 \frac{
[\hv-\nu(\alpha)\gv]^{H}K_{\Xv}[\hv-\nu(\alpha)\gv]+1+|\nu(\alpha)|^2-\nu^{*}(\alpha)\rho_{\rm
o}-\rho_{\rm
o}^{*}\nu(\alpha)}{1-|\rho_{\rm o}|^2} \notag\\
& = \log_2 \frac{\alpha P \left|[\hv-\rho_{\rm o} \gamma_1(\alpha)\gv]^{H}
\ev_1\right|^2 +1 + |\rho_{\rm o}|^2 \gamma_1^2(\alpha) - 2|\rho_{\rm o}|^2
\gamma_1(\alpha)}{1-|\rho_{\rm o}|^2}. \label{eq:m0}
\end{align}
Based on the definition of $\rho_{\rm o}$ in (\ref{eq:def-rhoo}), we have
\begin{align}
\hv- \rho_{\rm o} \gamma_1(\alpha)\gv &=\hv-
\gamma_1(\alpha) \frac{\gv\gv^{H}\ev_1} {\hv^{H}\ev_1}\notag\\
&=\left[\frac{\hv\hv^{H}- \gamma_1(\alpha)
\gv\gv^{H}}{\hv^{H}\ev_1}\right]\ev_1. \label{eq:apx-m13}
\end{align}
Hence,
\begin{align}
\left|[\hv- \rho_{\rm o} \gamma_1(\alpha)\gv]^{H}\ev_1\right|^2
&=\frac{\bigl[|\hv^{H}\ev_1|^2- \gamma_1(\alpha)
|\gv^{H}\ev_1|^2\bigr]^2}{|\hv^{H}\ev_1|^2} \notag\\
&=\left[|\hv^{H}\ev_1|^2- \gamma_1(\alpha) |\gv^{H}\ev_1|^2\right] [1-\gamma_1(\alpha) |\rho_{\rm o}|^2]\notag\\
&=\left[\frac{\gamma_1(\alpha)-1}{\alpha P} \right] [1-\gamma_1(\alpha)
|\rho_{\rm o}|^2]  \label{eq:apx-m1}
\end{align}
where the last step of (\ref{eq:apx-m1}) follows from the definition of
$\gamma_1(\alpha)$ in (\ref{eq:def-gm-1}). Substituting (\ref{eq:apx-m1}) into
(\ref{eq:m0}), we obtain
\begin{align}
f_1(\rho_{\rm o},K_{\Xv})  &\le  \log_2 \frac{[\gamma_1(\alpha)-1] [1-
\gamma_1(\alpha)|\rho_{\rm o}|^2] +1 + |\rho_{\rm o}|^2 \gamma_1^2(\alpha) -
2|\rho_{\rm o}|^2
\gamma_1(\alpha)}{1-|\rho_{\rm o}|^2}\notag\\
&= \log_2 \frac{\gamma_1(\alpha)-|\rho_{\rm o}|^2
\gamma_1(\alpha)}{1-|\rho_{\rm o}|^2}\notag\\
 &=  \log_2 \gamma_1(\alpha). \label{eq:m2}
\end{align}

\subsubsection{Bound on $f_2(\rho_{\rm o},K_{\Xv})$}

Here, we prove that if $K_{\Xv}\in \Lc(\alpha)$, then
\begin{align}
f_1(\rho_{\rm o},K_{\Xv}) \le \log_2 \gamma_2(\alpha)
\end{align}
where $\gamma_2(\alpha)$ is the largest generalized eigenvalue of the pencil
(\ref{eq:pencil2}). In fact, the smallest generalized eigenvalue of the pencil
(\ref{eq:pencil2}) is $\gamma_1(\alpha)/\lambda_1$. This result is described in
the following lemma.

\begin{lemma} \label{lem:small}
$\gamma_1(\alpha)/\lambda_1$ and $\ev_1(\alpha)$ are the smallest generalized
eigenvalue and the corresponding normalized eigenvector of the pencil
\begin{align}
\left(I+\frac{(1-\alpha) P}{1+\alpha P |\gv^{H}\ev_1|^2}\gv\gv^{H}, \;
I+\frac{(1-\alpha) P}{1+ \alpha P |\hv^{H}\ev_1|^2}\hv \hv^{H}\right)
\end{align}
where $\lambda_1$ and $\ev_1(\alpha)$ are defined in (\ref{eq:eig-def1}), and
$\gamma_1(\alpha)$ is defined in (\ref{eq:def-gm-1}).
\end{lemma}
\begin{proof}
We provide the proof in Appendix~\ref{app:sec5}.
\end{proof}
Based on the property of generalized eigenvalues (see Appendix~\ref{app:geig}),
Lemma~\ref{lem:small} implies that
\begin{align}
&  &\ev_1^{H}\left[I+\frac{(1-\alpha)
P}{1+\alpha P
|\gv^{H}\ev_1|^2}\gv\gv^{H}\right]\cv_2(\alpha) &=0 &\label{eq:e1ooc2}\\
&\text{and}  &\ev_1^{H}\left[I+\frac{(1-\alpha) P}{1+\alpha P
|\hv^{H}\ev_1|^2}\hv\hv^{H}\right]\cv_2(\alpha) &=0 &
\end{align}
where $\cv_2(\alpha)$ is the normalized eigenvector of the pencil
(\ref{eq:pencil2}) corresponding to $\gamma_2(\alpha)$. Hence,
\begin{align}
&  &\frac{(1-\alpha) P}{1+\alpha P
|\gv^{H}\ev_1|^2}  \ev_1^{H}\gv \gv^{H}\cv_2(\alpha) &=-\ev_1^{H}\cv_2(\alpha) & \label{eq:e1c2}\\
&\text{and}  &\frac{(1-\alpha) P}{1+\alpha P |\hv^{H}\ev_1|^2} \ev_1^{H}\hv
\hv^{H}\cv_2(\alpha) &=-\ev_1^{H}\cv_2(\alpha). &
\end{align}
By combining the definitions of $\rho_{\rm o}$ in (\ref{eq:def-rhoo}) and
$\gamma_1(\alpha)$ in (\ref{eq:def-gm-1}), we obtain
\begin{align} \rho_{\rm o}
=\frac{\gv^{H}\ev_1}{\hv^{H}\ev_1}
=\frac{1}{\gamma_1(\alpha)}\left[\frac{\hv^{H}\cv_2(\alpha)}{\gv^{H}\cv_2(\alpha)}\right]^{*}.
\label{eq:def-rho2}
\end{align}

We now establish the relationship between $\gamma_1(\alpha)$ and
$\gamma_2(\alpha)$  based on (\ref{eq:def-rho2}) in the following lemma.
\begin{lemma} \label{lem:orth}
For any $\alpha \in [0,1]$,
\begin{align}
&          &  \frac{\gv-\rho_{\rm o}^{*} \gamma_2(\alpha)\hv} {|\gv-\rho_{\rm
o}^{*}\gamma_2(\alpha)\hv|^2}&=\cv_2(\alpha)&  \label{eq:c2}\\
&\text{and}& [\hv- \rho_{\rm o} \gamma_1(\alpha)\gv]^{H}[\gv- \rho_{\rm o}^{*}
\gamma_2(\alpha)\hv]&=0 &
\end{align}
where $\gamma_1(\alpha)$ is defined in (\ref{eq:def-gm-1}), and
$\gamma_2(\alpha)$ and $\cv_2(\alpha)$ are the largest generalized eigenvalue
and the corresponding normalized eigenvector of the pencil (\ref{eq:pencil2}).
\end{lemma}
\begin{proof}
We provide the proof in Appendix~\ref{app:sec5}.
\end{proof}
Let $\cv_1(\alpha)$ denote the normalized vector of $\hv-\rho_{\rm o}
\gamma_1(\alpha)\gv$, i.e.,
\begin{align}
\cv_1(\alpha)\triangleq\frac{\hv-\rho_{\rm o}
\gamma_1(\alpha)\gv}{|\hv-\rho_{\rm o} \gamma_1(\alpha)\gv|^2}.
\end{align}
Note that Lemma~\ref{lem:orth} implies that $\cv_1(\alpha)$ and $\cv_2(\alpha)$
are orthogonal. Moreover, since the input covariance matrix $K_{\Xv}$ is
Hermitian and positive semidefinite, we obtain
\begin{align}
\cv_1^{H}(\alpha)K_{\Xv}\cv_1(\alpha)+\cv_2^{H}(\alpha)K_{\Xv}\cv_2(\alpha) \le
\tr(K_{\Xv})=P.
\end{align}
Hence, for a given $K_{\Xv}\in\Lc(\alpha)$, we have
\begin{align}
\cv_2^{H}(\alpha)K_{\Xv}\cv_2(\alpha) &\le P -\alpha
P|\cv_1^{H}(\alpha)\ev_1|^2 \notag\\
&= (1-\alpha)P +\alpha P|\cv_2^{H}(\alpha)\ev_1|^2. \label{eq:pc2}
\end{align}
Inserting (\ref{eq:c2}) into (\ref{eq:pc2}), we obtain
\begin{align}
[\gv-\rho_{\rm o}^{*}\gamma_2(\alpha)\hv]^{H}K_{\Xv} [\gv-\rho_{\rm
o}^{*}\gamma_2(\alpha)\hv] \le (1-\alpha)P\zeta(\alpha) +\alpha P \eta(\alpha)
\end{align}
where
\begin{align}
&            & \zeta(\alpha)&\triangleq|\gv-\rho_{\rm o}^{*}\gamma_2(\alpha)\hv)|^2& \label{eq:zta}\\
& \text{and} & \eta(\alpha) &\triangleq |\gv- \rho_{\rm o}^{*}
\gamma_2(\alpha)\hv|^2 |\ev_1^{H}\cv_2(\alpha)|^2. & \label{eq:eta}
\end{align}
In Appendix~\ref{app:sec5}, we prove the following equality
\begin{align}
(1-\alpha)P\zeta(\alpha) +\alpha P \eta(\alpha)=[\gamma_2(\alpha)-1][1-
\gamma_2(\alpha)|\rho_{\rm o}|^2]. \label{eq:fm6}
\end{align}

Next, we consider the bound on $f_2(\rho_{\rm o},K_{\Xv})$. Let
\begin{align*}
\mu(\alpha)=\rho_{\rm o}^{*} \gamma_2(\alpha).
\end{align*}
For a given $K_{\Xv}\in\Lc(\alpha)$, the definition (\ref{eq:f2-def}) implies
that
\begin{align}
f_2(\rho_{\rm o},K_{\Xv}) & \le \log_2 \frac{
[\gv-\mu(\alpha)\hv]^{H}K_{\Xv}[\gv-\mu(\alpha)\hv]+1+|\mu(\alpha)|^2-\mu^{*}(\alpha)\rho_{\rm
o}-\rho_{\rm
o}^{*}\mu(\alpha)}{1-|\rho_{\rm o}|^2} \notag\\
& \le \log_2 \frac{(1-\alpha)P\zeta(\alpha)+\alpha P\eta(\alpha) +1 +
|\rho_{\rm o}|^2 \gamma_2^2(\alpha) - 2|\rho_{\rm o}|^2
\gamma_2(\alpha)}{1-|\rho_{\rm o}|^2}. \label{eq:fm0}
\end{align}
Now, substituting (\ref{eq:fm6}) into (\ref{eq:fm0}), we obtain
\begin{align}
f_2(\rho_{\rm o},K_{\Xv}) & \le \log_2 \frac{[\gamma_2(\alpha)-1][1-
\gamma_2(\alpha)|\rho_{\rm o}|^2] +1 + |\rho_{\rm o}|^2 \gamma_2^2(\alpha) -
2|\rho_{\rm o}|^2 \gamma_2(\alpha)}{1-|\rho_{\rm o}|^2}\notag\\
&=\log_2 \gamma_2(\alpha). \label{eq:fm7}
\end{align}

Finally, Combining (\ref{eq:lset}), (\ref{eq:m2}) and (\ref{eq:fm7}), we have
the desired result:
\begin{align}
\bigcup_{\tr(K_{\Xv})\le P} {\Rc}_{\rm O}(\rho, K_{\Xv})\subseteq \bigcup_{0\le
\alpha \le 1} \Rc^{\rm MG}(\alpha).
\end{align}

\section{Numerical Examples} \label{sec:ex}

In this section, we study two numerical examples to illustrate the secrecy
capacity region of the MGBC-CM. For simplicity, we assume that the Gaussian BC
has real input and output alphabets and the channel attenuation vectors $\hv$
and $\gv$ are real too. Under this condition, all calculated rate values are
divided by $2$.

\begin{example}
\begin{figure}[t]
 \centerline{\includegraphics[width=0.6\linewidth,draft=false]{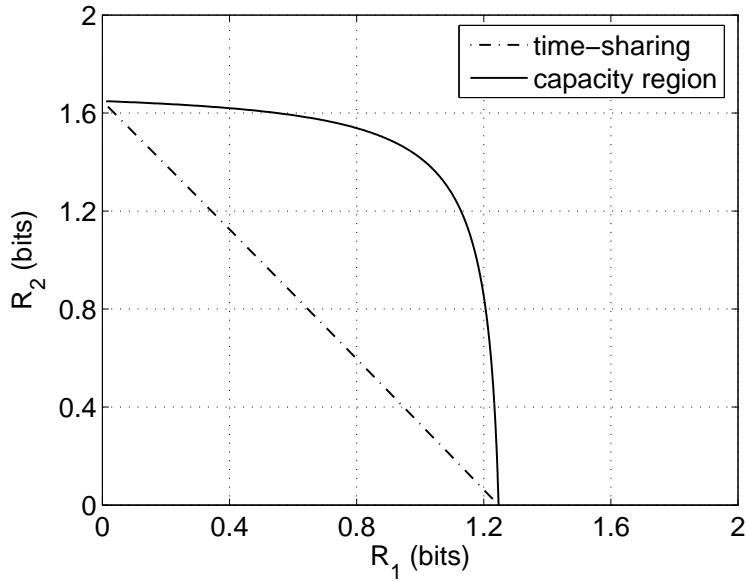}}
\caption{Comparison of the Sato-type outer bound and secrecy rate regions
achieved by time-sharing and simplified DPC schemes for the example MGBC-CM in
(\ref{eq:ex1})}
  \label{fig:sim1}
\end{figure}
In the first example, we consider the following MGBC-CM
\begin{align}
\left[\begin{matrix}y_1 \\ y_2 \end{matrix}\right] &= \left[\begin{matrix} 1.5
&0 \\ 1.801 & 0.871 \end{matrix}\right] \left[\begin{matrix}x_1 \\ x_2
\end{matrix}\right] +\left[\begin{matrix}z_1 \\ z_2
\end{matrix}\right]\label{eq:ex1}
\end{align}
where $\hv=[1.5, 0]^{T}$, $\gv=[1.801, 0.872]^{T}$, and the total power
constraint is set to $P=10$. Fig.~\ref{fig:sim1} illustrates the secrecy
capacity region for the channel (\ref{eq:ex1}). We observe that even though
each component of the attenuation vector $\hv$ (imposed on user 1) is strictly
less than the corresponding component of $\gv$ (imposed on user 2), both users
can achieve positive rates simultaneously under the information-theoretic
secrecy requirement.
\end{example}

\begin{example}
In the second example, we consider the MGBC-CM as follows
\begin{align}
\left[\begin{matrix}y_1 \\ y_2 \end{matrix}\right]&= \left[\begin{matrix} 1.414
& 1.414 \\ 0.4  &   1.959 \end{matrix}\right] \left[\begin{matrix}x_1 \\ x_2
\end{matrix}\right] +\left[\begin{matrix}z_1 \\ z_2
\end{matrix}\right]\label{eq:ex2}
\end{align}
where $\hv=[1.414, 1.414]^{T}$, $\gv=[0.4, 1.959]^{T}$, and the total power
$P=10$. The secrecy capacity region of the channel (\ref{eq:ex2}) is calculated
and depicted in Fig.~\ref{fig:sim2}.
\begin{figure}[t]
 \centerline{\includegraphics[width=0.6\linewidth,draft=false]{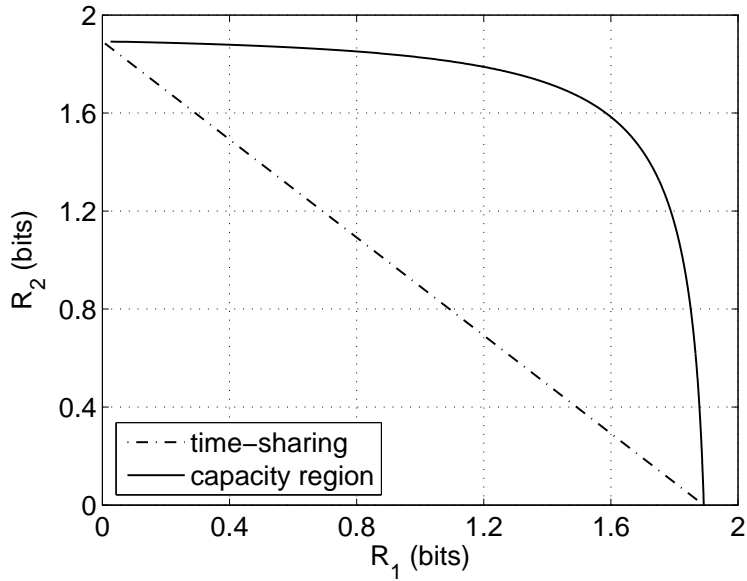}}
\caption{Comparison of the Sato-type outer bound and secrecy rate regions
achieved by time-sharing and simplified DPC schemes for the example MGBC-CM in
(\ref{eq:ex2})}
  \label{fig:sim2}
\end{figure}
\end{example}

Moreover, we compare the secrecy capacity region with the secrecy rate region
achieved by the time-sharing scheme (indicated by the dash-dot line). The
time-sharing refers to the scheme in which the transmitter sends the
confidential message $W_1$  with total power $P_1$ during a fraction $\tau_1$
of time, and sends the confidential message $W_2$ with total power $P_2$ during
a fraction $\tau_2$ of time, where
$$\tau_1+\tau_2=1 \quad  \text{and} \quad \tau_1 P_1+\tau_2 P_2=P.$$
Note that in each time fraction, the MGBC-CM reduces to a Gaussian MISO wiretap
channel. Using such time-sharing, the rate pair $\left(\frac{\tau_1}{2} \log_2
\lambda_1(P_1),  \frac{\tau_2}{2} \log_2 \lambda_2(P_2)\right)$ is achievable,
where $\lambda_1(P_1)$ and $\lambda_2(P_2)$ are the largest generalized
eigenvalues of the pencil $(I+P_1\hv \hv^{H}, I+P_1\gv \gv^{H})$ and the pencil
$(I+P_2\gv \gv^{H}, I+P_2\gv \gv^{H})$, respectively. Both Fig.~\ref{fig:sim1}
and Fig.~\ref{fig:sim2} demonstrate that the time-sharing scheme is strictly
suboptimal for providing the secrecy capacity region.

\section{Conclusion} \label{sec:con}

In this paper, we have investigated the secrecy capacity region of a generally
non-degraded Gaussian BC with confidential messages for two users, where the
transmitter has $t$ antennas and each user has a single antenna. For this
model, we have proposed a secret dirty-paper coding scheme and introduced a
computable Sato-type outer bound. Furthermore, we have proved that the boundary
of the secret dirty-paper coding rate region is consistent with the Sato-type
outer bound for the multiple-antenna Gaussian BC, and hence, we have obtained
the secrecy capacity region for the MGBC-CM.

Unlike the single-antenna Gaussian BC-CM case, in which only the superior user
can obtain confidential information at a positive secrecy rate, our result has
illustrated that both users can achieve strictly positive rates with
information-theoretic secrecy through a multiple-antenna Gaussian BC if
attenuation vectors imposed on user 1 and user 2 are linear independent.
Therefore, it becomes more practical and more attractive to achieve
information-theoretic secrecy in wireless networks by employing multiple
transmit-antennas at the physical layer.

%The relevance of this work lies in a number of recent information-theoretic
%results on secret communication by other authors, including the study on MISO
%wiretap channel in \cite{Li:CISS:07,Wornell:ISIT:07,Ulukus:ISIT:07} and the
%work on the discrete memoryless BC-CM in \cite{Liu07it}.
%
%The achievability proof have suggested the secret DPC scheme for the MGBC-CM.
%However, constructing practical wiretap codes that can achieve the derived
%rates is still challenging. Code constructions for a wiretap channel have been
%studied in \cite{Thangaraj:ARXIV:05,Bloch:ISIT:06,Liu:ITW:07}.

\appendices

\section{The Generalized Eigenvalue and Rayleigh Quotient Problem}
\label{app:geig}

A generalized eigenvalue problem is to determine the nontrivial solutions of
the equation
\begin{align}
A\ev=\lambda B \ev \label{eq:geig-d}
\end{align}
where $A$ and $B$ are matrices and $\lambda$ is a scalar. The values of
$\lambda$ that satisfy (\ref{eq:geig-d}) are the generalized eigenvalues and
the corresponding vectors of $\ev$ are the generalized eigenvectors.

In particular, if $A$ is Hermitian and $B$ is Hermitian and positive definite,
then we have the following properties of $A\ev=\lambda B \ev$:
\begin{enumerate}
  \item The generalized eigenvalues $\lambda_i$ are real.
  \item The eigenvectors are ``$B$-orthogonal'', i.e.,
  \begin{align}
  \ev_i^{H} B \ev_j = 0 \quad \text{for}~i\neq j.
  \end{align}
  \item Similarly,
  \begin{align}
  \ev_i^{H} A \ev_j = \lambda_j \ev_i^{H} A \ev_j =0 \quad \text{for}~i\neq j.
  \end{align}
\end{enumerate}

Next, we describe the well-known {\it Rayleigh's quotient} \cite{Gilbert} as
follows.
\begin{theorem} \label{thm:Q} (see \cite{Gilbert})
Let $r(\cv)$ be the Rayleigh's quotient defined as
\begin{align}
r(\cv)\triangleq \frac{\cv^{H} A  \cv}{\cv^{H} B  \cv}.
\end{align}
where $A$ is Hermitian and $B$ is Hermitian and positive definite. The quotient
$R(\cv)$ is maximized by the eigenvector $\ev_{\max}$ corresponding to the
largest generalized eigenvalue $\lambda_{\max}$ of the pencil $(A,B)$:
\begin{align}
\max_{\cv} R(\cv)= \frac{\ev_{\max}^{H} A  \ev_{\max}}{\ev_{\max}^{H} B
\ev_{\max}}&=\lambda_{\max}
\end{align}
and $R(\cv)$ is minimized by the eigenvector $\ev_{\min}$ corresponding to the
smallest generalized eigenvalue $\lambda_{\min}$ of the pencil $(A,B)$:
\begin{align}
\min_{\cv} R(\cv)= \frac{\ev_{\min}^{H} A  \ev_{\min}}{\ev_{\min}^{H} B
\ev_{\min}}&=\lambda_{\min}.
\end{align}
\end{theorem}
The proof of Theorem~\ref{thm:Q} can be found in \cite[Chapter~6]{Gilbert}. Now
we prove Lemma~\ref{lem:geig} based on the Rayleigh's quotient principle.

\begin{proof}({\bf Lemma~\ref{lem:geig}})
Since both $(I+P\hv\hv^{H})$ and $(I+P\gv\gv^{H})$ are Hermitian and positive
definite matrices, the definition of $\lambda_1$ and Theorem~\ref{thm:Q} imply
that
\begin{align}
\lambda_1= \max_{\cv} \frac{\cv^{H}(I+P\hv\hv^{H}) \cv}{\cv^{H}
(I+P\gv\gv^{H})\cv}.
\end{align}
We consider a unit vector $\cv_0$ that is orthogonal with the vector $\gv$,
i.e., $$\cv_0^{H}\cv_0=1 \quad \text{and} \quad \cv_0^{H}\gv=0.$$ Now, we have
\begin{align}
\lambda_1 &\ge \frac{\cv_0^{H}(I+P\hv\hv^{H}) \cv_0}{\cv_0^{H}
(I+P\gv\gv^{H})\cv_0} \notag\\
&=1+P|\cv_0^{H}\hv|^2 . \label{eq:lg1}
\end{align}
This implies that $\lambda_1\ge1$. Furthermore, when $\hv$ and $\gv$ are linear
independent, there exists a unit vector $\cv_0$ so that
\begin{align}
\cv_0^{H}\gv=0 \quad \text{and}  \quad \cv_0^{H}\hv>0 . \label{eq:lg2}
\end{align}
Substituting (ref{eq:lg2}) into (\ref{eq:lg1}), we obtain
$\lambda_1>1$. By using the same approach, we can show that $\lambda_2\ge 1$,
and, in particular, $\lambda_2> 1$ when $\hv$ and $\gv$ are linear independent.
\end{proof}

\section{Proof of Corollary~\ref{cor:maxr1}} \label{app:sec3}

\begin{proof} %({\bf Corollary~\ref{cor:maxr1}})
Theorem~\ref{thm:GBC} demonstrates that for a given $\alpha\in[0,1]$, the
maximum achievable secrecy rate of user~1 is $\gamma_1(\alpha)$. This implies
that
\begin{align}
R_{1,\max} =\max_{0\le \alpha\le1}\gamma_1(\alpha).
\end{align}
We also notice that $\lambda_1=\gamma_1(1)$. Hence, it is sufficient to show
that $\gamma_1(\alpha)$ is a nondecreasing function on an interval $[0,1]$.

Let
\begin{align}
\kappa(\alpha)\triangleq\frac{d \gamma_1(\alpha)}{d \alpha}.
\end{align}
Based on the definition of $\gamma_1(\alpha)$ in (\ref{eq:def-gm-1}), we can
write
\begin{align}
\kappa(\alpha)&=\frac{P|\hv^{H}\ev_1|^2 (1+ \alpha P
|\gv^{H}\ev_1|^2)-(1+\alpha P|\hv^{H}\ev_1|^2)P |\gv^{H}\ev_1|^2}{\left(1+
\alpha P |\gv^{H}\ev_1|^2\right)^2} \notag\\
&=\frac{P(|\hv^{H}\ev_1|^2 - |\gv^{H}\ev_1|^2)}{\left(1+ \alpha P
|\gv^{H}\ev_1|^2\right)^2}
\end{align}
Now, the definitions of $\lambda_1$ and $\ev_1$ in (\ref{eq:eig-def1}) imply
that
\begin{align}
|\hv^{H}\ev_1|^2- \lambda_1|\gv^{H}\ev_1|^2&= \frac{\lambda_1-1}{P}.
\end{align}
Moreover, Since $\lambda_1\ge 1$ (see Lemma~\ref{lem:geig}), we obtain
\begin{align}
\kappa(\alpha)\ge0,
\end{align}
and hence, $\gamma_1(\alpha)$ is a nondecreasing function on an interval
$[0,1]$. Therefore, we have
$$R_{1,\max}=\max_{0\le\alpha\le1} \gamma_1(\alpha)=\lambda_1.$$
\end{proof}

\section{Section~\ref{sec:in} Derivations} \label{app:sec4}

\subsubsection*{\rm ({\bf Double-Binning Scheme})}

By contrast with the classical DPC scheme, the secret DPC scheme is based on
the double-binning code structure as follows. Let
\begin{align}
R^{\star}_1 &=I(\Vv_1;Y_2|\Vv_2), \quad R^{\star}_2=I(\Vv_2;Y_1|\Vv_1) \quad
\text{and} \quad R^{\ddag}=I(\Vv_1;\Vv_2).
\end{align}
Generate $2^{n(R_k+R^{\star}_k+R^{\ddag})}$ codewords $\vv_k^{n}(w_k,j_k,l_k)$,
$w_k=1,2,\dots,2^{R_k}$, $j_k=1,2,\dots,2^{R^{\star}_k}$,
$l_k=1,2,\dots,2^{R^{\ddag}}$, independently at random according to $p(\vv_k)$.
Based on the labeling, we partition the codebook $\{\vv_k^{n}(w_k,j_k,l_k)\}$
into $2^{nR_k}$ {\it bins}, where bin~$w_k$ represents the message index $w_k$.
We further divide bin~$w_k$ into $2^{nR^{\star}_k}$ {\it sub-bins}. Each
sub-bin~$(w_k,j_k)$ contains $2^{nR^{\ddag}}$ codewords.

To send the message pair $(w_1,w_2)$, the transmitter employs a joint
stochastic encoder. We first randomly select a sub-bin~$(w_1,j_1)$ from the
bin~$w_1$ and randomly choose a codeword $\vv_1^{n}(w_1,j_1,l_1)$ from the
sub-bin~$(w_1,j_1)$. Next, we randomly select a sub-bin~$(w_2,j_2)$ from the
bin $w_2$ and find a codeword $\vv_2^{n}(w_2,j_2,l_2)$ in the
sub-bin~$(w_2,j_2)$ so that the sequences $\vv_1^{n}(w_1,j_1,l_1)$ and
$\vv_2^{n}(w_2,j_2,l_2)$ are jointly typical with respective to
$p(\vv_1,\vv_2)$. Since each sub-bin contains $2^{nI(\Vv_1;\Vv_2)}$ codewords,
the encoding is successful with probability close to $1$ as long as $n$ is
large. Finally, we generate the channel input sequence $\xv^{n}(w_1,w_2)$
according to the mapping $p(\xv|\vv_1,\vv_2)$.

\bigskip

\begin{proof}({\bf Lemma~\ref{lem:GBCin}})
We first check the power constraint. Since $\Uv_1$ and $\Uv_2$ are independent
and
$$\Xv=\Uv_1+\Uv_2,$$ the covariance matrices $K_{\Uv_1}$ and $K_{\Uv_2}$
satisfy
\begin{align}
\tr(K_{\Uv_1}+K_{\Uv_2})=\tr(K_{\Xv})\le P.
\end{align}

Following from \cite[Theorem~1]{Yu:IT:04} and using the setting in
(\ref{eq:rvs}), we can immediately obtain the well-known \emph{successive
dirty-paper encoding} result:
\begin{align}
I(\Vv_1;Y_1)-I(\Vv_1;\Vv_2)
%&= h(\Vv_1|\Vv_2)- h(\Vv_1|Y_1) \notag\\
%&= h(\Uv_1)- h(\Uv_1|\hv^{H}\Uv_1+Z_1) \notag\\
&= I(\Uv_1; \hv^{H}\Uv_1+Z_1) \notag\\
&=\log_2 (1+\hv^{H} K_{\Uv_1} \hv). \label{eq:pf1}
\end{align}
Since $\Vv_2=\Uv_2$ is independent of $\Uv_1$ and $\Vv_1=\Uv_1+\bv\hv^{H}
\Uv_2$, we obtain
\begin{align}
I(\Vv_1;Y_2|\Vv_2)&=I(\Uv_1+\bv\hv^{H} \Uv_2; Y_2|\Uv_2)\notag\\
&=I(\Uv_1;Y_2|\Uv_2) \notag\\
&=\log_2(1+\gv^{H} K_{\Uv_1} \gv). \label{eq:pf2}
\end{align}
Combining (\ref{eq:BC-IN-R1}), (\ref{eq:pf1}) and (\ref{eq:pf2}), we have
\begin{align}
R_1 & \le I(\Vv_1;Y_1)-I(\Vv_1;\Vv_2)-I(\Vv_1;Y_2|\Vv_2) \notag\\
&=\log_2 \frac{1+\hv^{H} K_{\Uv_1} \hv}{1+\gv^{H} K_{\Uv_1} \gv}.
\label{eq:pf3}
\end{align}
Moreover, we can compute
\begin{align}
I(\Vv_2;Y_2)&= h(Y_2)-h(Y_2|\Vv_2)\notag\\
%&= h(\gv^{H}\Xv+Z_2)-h(\gv^{H}\Xv+Z_2|\Uv_2)\notag\\
&= h(\gv^{H}(\Uv_1+\Uv_2)+Z_2)-h(\gv^{H}\Uv_1+Z_2)\notag\\
&= \log_2 \frac{1+\gv^{H} (K_{\Uv_1}+K_{\Uv_2}) \gv}{1+\gv^{H} K_{\Uv_1} \gv}.
\label{eq:pf4}
\end{align}
and
\begin{align}
I(\Vv_2;Y_1|\Vv_1)+I(\Vv_1;\Vv_2)&=I(\Vv_1,\Vv_2;Y_1)-[I(\Vv_1;Y_1)-I(\Vv_1;\Vv_2)]\notag\\
&=\log_2\frac{1+\hv^{H} (K_{\Uv_1}+K_{\Uv_2}) \hv}{1+\hv^{H} K_{\Uv_1}\hv}.
\label{eq:pf5}
\end{align}
Substituting (\ref{eq:pf4}) and (\ref{eq:pf5}) into (\ref{eq:BC-IN-R2}), we
obtain that
\begin{align}
R_2 & \le I(\Vv_2;Y_2)-I(\Vv_1;\Vv_2)-I(\Vv_2;Y_1|\Vv_1) \notag\\
&=\log_2 \frac{1+\gv^{H} (K_{\Uv_1}+K_{\Uv_2}) \gv}{1+\gv^{H} K_{\Uv_1} \gv}
-\log_2\frac{1+\hv^{H} (K_{\Uv_1}+K_{\Uv_2}) \hv}{1+\hv^{H}
K_{\Uv_1}\hv}\notag\\
&=\log_2 \frac{1+\gv^{H} (K_{\Uv_1}+K_{\Uv_2}) \gv}{1+\hv^{H}
(K_{\Uv_1}+K_{\Uv_2}) \hv} +\log_2\frac{1+\hv^{H} K_{\Uv_1}\hv}{1+\gv^{H}
K_{\Uv_1} \gv}.\label{eq:pf6}
\end{align}
Applying Lemma~\ref{lem:InBC} with bounds (\ref{eq:pf3}) and (\ref{eq:pf6}), we
have the desired result.
\end{proof}

\bigskip

\begin{proof}({\bf Equation (\ref{eq:sdpc-r2})})
For convenience,we define
\begin{align}
d(K_{\Uv_1}, K_{\Uv_2})&\triangleq \frac{[1+\gv^{H} (K_{\Uv_1}+K_{\Uv_2}) \gv]
[1+\hv^{H} K_{\Uv_1}\hv]}{[1+\hv^{H} (K_{\Uv_1}+K_{\Uv_2}) \hv][1+\gv^{H}
K_{\Uv_1} \gv]} \notag\\
&=\left[1+\frac{\gv^{H} K_{\Uv_2} \gv}{1+\gv^{H} K_{\Uv_1} \gv}\right]
\left[1+\frac{\hv^{H}K_{\Uv_2} \hv}{1+\hv^{H} K_{\Uv_1}\hv}\right]^{-1}.
\end{align}
Since $\cv_2^{H}(\alpha)\cv_2(\alpha)=1$, we substitute (\ref{eq:ku12}) into
$d(K_{\Uv_1}, K_{\Uv_2})$ and obtain
\begin{align}
d(K_{\Uv_1}, K_{\Uv_2}) &=\left[1+\frac{(1-\alpha) P  \gv^{H}
\cv_2(\alpha)\cv_2^{H}(\alpha) \gv}{1+\alpha P |\gv^{H}\ev_1|^2}\right]
\left[1+\frac{(1-\alpha) P \hv^{H}\cv_2(\alpha)\cv_2^{H}(\alpha) \hv}{1+ \alpha P|\hv^{H}\ev_1|^2}\right]^{-1}\notag\\
&=\frac{\displaystyle{\cv_2^{H}(\alpha)\cv_2(\alpha)+\frac{(1-\alpha) P
\cv_2^{H}(\alpha)\gv\gv^{H}\cv_2(\alpha)}{1+\alpha P
|\gv^{H}\ev_1|^2}}}{\displaystyle{
\cv_2^{H}(\alpha)\cv_2(\alpha)+\frac{(1-\alpha) P
\cv_2^{H}(\alpha)\hv\hv^{H}\cv_2(\alpha)}{1+ \alpha P
|\hv^{H}\ev_1|^2}}} \notag\\
&=\frac{\cv_2^{H}(\alpha)\left[\displaystyle{I+\frac{(1-\alpha)
P\gv\gv^{H}}{1+\alpha P|\gv^{H}\ev_1|^2}}\right]\cv_2(\alpha)}
{\cv_2^{H}(\alpha)\left[\displaystyle{ I+\frac{(1-\alpha) P \hv\hv^{H}}{1+
\alpha P |\hv^{H}\ev_1|^2}}\right]\cv_2(\alpha)}.
\end{align}
Note that $\gamma_2(\alpha)$ and $\cv_2(\alpha)$ are the largest generalized
eigenvalue and the corresponding normalized eigenvector of the pencil
(\ref{eq:pencil2}), i.e.,
\begin{align*}
\left(I+\frac{(1-\alpha) P}{1+\alpha P |\gv^{H}\ev_1|^2}\gv\gv^{H}\right)
\cv_2(\alpha) = \gamma_2(\alpha) \left( I+\frac{(1-\alpha) P}{1+ \alpha P
|\hv^{H}\ev_1|^2}\hv \hv^{H}\right) \cv_2(\alpha).
\end{align*}
Hence, we have
\begin{align}
d_2(K_{\Uv_1}, K_{\Uv_2}) =\gamma_2(\alpha).
\end{align}
\end{proof}

\section{Section~\ref{sec:out} Derivations} \label{app:sec5}

\begin{proof}({\bf Lemma~\ref{lem:sam}})
It is sufficient to show that the error probability $P_e^{(n)}$ and the
equivocations $H(W_2|Y_1^{n})$ and $H(W_2|Y_1^{n})$ are the same for the
channels $p_{\Yt_1,\Yt_2|\Xv} \in \Pc$ when we use the same codebook and
encoding schemes. We note that
\begin{align}
P_e^{(n)} = \max \bigl\{P_{e,1}^{(n)}, P_{e,2}^{(n)}\bigr\} \le
P_{e,1}^{(n)}+P_{e,2}^{(n)}.
\end{align}
Hence, $P_e^{(n)}$ is small if and only if both $P_{e,1}^{(n)}$ and
$P_{e,2}^{(n)}$ are small. However, for given codebook and encoding scheme
$p(\xv^{n}|w_1,w_2)$, the decoding error probability $P_{e,k}^{(n)}$ and the
equivocation rate at user $k$ depend only on the marginal channel probability
density $p_{\Yt_k|\Xv}$. Therefore, the same code and encoding scheme for any
$p_{\Yt_1,\Yt_2|\Xv} \in \Pc$ gives the same $P_e^{(n)}$ and equivocation
rates. This concludes the proof.
\end{proof}

\bigskip

\begin{proof}({\bf Theorem~\ref{thm:out1}})
Here we prove Theorem~\ref{thm:out1} and derive the outer bound for $R_1$. The
outer bound for $R_2$ follows by symmetry.

The secrecy requirement (\ref{eq:equiv}) implies that
\begin{align}
nR_1= H(W_1) &\le H(W_1|Y_2^{n})+n\epsilon  \label{eq:r1}.
\end{align}
On the other hand, Fano's inequality and $P_e\le \epsilon$ imply that
\begin{align}
H(W_1|Y_1^{n}) &\le \epsilon \log(M_1-1)+h(\epsilon) \triangleq n\delta_1.
\label{eq:sd1}
\end{align}
where $h(x)$ is the binary entropy function. Based on (\ref{eq:r1}) and
(\ref{eq:sd1}), we have
\begin{align}
nR_1& \le H(W_1|Y_2^{n})+n\epsilon \notag\\
&\le H(W_1|Y_2^{n})-H(W_1|Y_1^{n})+n(\delta_1+\epsilon) \notag\\
&\le H(W_1|Y_2^{n})-H(W_1|Y_1^{n},Y_2^{n})+n(\delta_1+\epsilon) \label{eq:cr01}\\
&= I(W_1;Y_1^{n}|Y_2^{n})+n(\delta_1+\epsilon) \label{eq:cr02}
\end{align}
where (\ref{eq:cr01}) follows from conditioning reducing entropy. Since $W_1
\rightarrow \Xv^{n} \rightarrow (Y_1^{n},Y_2^{n})$ forms a Markov chain, we can
further bound (\ref{eq:cr02}) as follows
\begin{align}
nR_1
&\le I(\Xv^{n};Y_1^{n}|Y_2^{n})+n(\delta_1+\epsilon) \notag\\
&\le \sum_{i=1}^{n} I(\Xv_i;Y_{1,i}|Y_{2,i})+n(\delta_1+\epsilon).
\end{align}

Finally, by applying Lemma~\ref{lem:sam}, we can replace $Y_1$ and $Y_2$ by
$\Yt_1$ and $\Yt_2$, respectively. Hence, we have the Sato-type outer bound on
$R_1$.
\end{proof}

\bigskip

\begin{proof} ({\bf Lemma~\ref{lem:outG}})
Here, we proof Lemma~\ref{lem:outG} based on the Sato-type outer bound in
Theorem~\ref{thm:out1}. For the Gaussian BC defined in (\ref{eq:miso-2}), the
upper bound (\ref{eq:BC-out-R1}) on $R_1$ can be rewritten as follows:
\begin{align}
I(\Xv;\Yt_1,\Yt_2)-I(\Xv;\Yt_2)&=h(\Yt_1,\Yt_2)-h(\Yt_1,\Yt_2|\Xv)-h(\Yt_2)+h(\Yt_2|\Xv)\notag\\
&=h(\Yt_1|\Yt_2)-[h(\Zt_1,\Zt_2)-h(\Zt_2)]\notag\\
&=h(\Yt_1|\Yt_2)-\log_2 (2\pi e)(1-|\rho|^2). \label{eq:aor1-1}
\end{align}
The first term of (\ref{eq:aor1-1}) can be further bounded as follows
\begin{align}
h(\Yt_1|\Yt_2)&=h(\Yt_1- \nu \Yt_2 |\Yt_2)\notag\\
&\le h(\Yt_1-\nu \Yt_2) \qquad \text{for any}~\nu \in \mathbb{C}
\label{eq:aor1-2}
\end{align}
where the inequality follows from removing conditioning. Moreover, the
maximum-entropy theorem \cite{Cover} implies that
\begin{align}
h(\Yt_1-\nu \Yt_2) &\le \log_2 (2\pi e)\bigl|\Var\bigl[\Yt_1-\nu
\Yt_2\bigr]\bigr| \notag\\
&=\log_2 (2\pi e) \bigl|\Var\bigl[(\hv-\nu\gv)^{H}\Xv\bigr]+\Var\bigl[\Zt_1-\nu \Zt_2\bigr]\bigr| \notag\\
&=\log_2 (2\pi e) \bigl[(\hv-\nu\gv)^{H} K_{\Xv}
(\hv-\nu\gv)+1+|\nu|^2-\nu^{*}\rho-\rho^{*}\nu \bigr]. \label{eq:upperb}
\end{align}
%Note that the bound (\ref{eq:aor1-2}) holds for any $\nu$.
Combining (\ref{eq:aor1-1}), (\ref{eq:aor1-2}) and (\ref{eq:upperb}), we obtain
the following upper bound:
\begin{align}
I(\Xv;\Yt_1,\Yt_2)-I(\Xv;\Yt_2)&\le \min_{\nu \in \mathbb{C}}  \log_2
\frac{(\hv-\nu\gv)^{H} K_{\Xv}
(\hv-\nu\gv)+1+|\nu|^2-\nu^{*}\rho-\rho^{*}\nu}{(1-|\rho|^2)}\\
&=f_1(\rho, K_{\Xv}). \label{eq:aor1}
\end{align}

Next we prove that for given $\rho$ and $K_{\Xv}$, the expression
$I(\Xv;\Yt_1,\Yt_2)-I(\Xv;\Yt_2)$ is maximized by Gaussian input distributions.
When $\Xv$ is a Gaussian random vector with zero-mean and covariance matrix
$K_{\Xv}$, the channel (\ref{eq:miso-2}) implies that $\Yt_1$ and $\Yt_2$ are
zero-mean Gaussian random variables. Choosing
\begin{align}
\nu=\nu_{\rm o}=\frac{\Cov \bigl[\Yt_1,\Yt_2\bigr]}{\Var\bigl[\Yt_2\bigr]}.
\end{align}
Note that
\begin{align}
\E\bigl[(\Yt_1- \nu_{\rm o} \Yt_2) \Yt_2^{*}\bigr]=\Cov
\bigl[\Yt_1,\Yt_2\bigr]-\nu_{\rm o} \Var\bigl[\Yt_2\bigr]=0,
\end{align}
the Gaussian random variables $\Yt_1- \nu_{\rm o} \Yt_2$ and $\Yt_2$ are
uncorrelated, and hence they are statistically independent. This implies that
\begin{align}
h(\Yt_1|\Yt_2)&= h(\Yt_1-\nu_{\rm o} \Yt_2) \notag\\
&=\log_2 (2\pi e) \bigl[(\hv-\nu_{\rm o}\gv)^{H} K_{\Xv} (\hv-\nu_{\rm
o}\gv)+1+|\nu_{\rm o}|^2-\nu_{\rm o}^{*}\rho-\rho^{*}\nu_{\rm o} \bigr].
\label{eq:aor1-3}
\end{align}
Inserting (\ref{eq:aor1-3}) into (\ref{eq:aor1-1}) we have
\begin{align}
I(\Xv;\Yt_1,\Yt_2)-I(\Xv;\Yt_2)&=\log_2 \frac{(\hv-\nu_{\rm o}\gv)^{H} K_{\Xv}
(\hv-\nu_{\rm o}\gv)+1+|\nu_{\rm o}|^2-\nu_{\rm o}^{*}\rho-\rho^{*}\nu_{\rm o}}{(1-|\rho|^2)} \notag\\
& \ge f_1(\rho, K_{\Xv}). \label{eq:aor1-4}
\end{align}
Bounds (\ref{eq:aor1}) and (\ref{eq:aor1-4}) imply that Gaussian input
distributions are optimal for the expression $I(\Xv;\Yt_1,\Yt_2)-I(\Xv;\Yt_2)$.

Following the same approach, we can prove that for given $\rho$ and $K_{\Xv}$,
Gaussian input distributions maximize the expression
$I(\Xv;\Yt_1,\Yt_2)-I(\Xv;\Yt_1)$, the upper bound (\ref{eq:BC-out-R2}) on
$R_2$. This lets us restrict attention to zero-mean Gaussian $\Xv$ with
covariance matrix $K_{\Xv}$. Now, bounds (\ref{eq:BC-out-R1}) and
(\ref{eq:BC-out-R2}) become
\begin{align}
&           & R_1 &\le f_1(\rho, K_{\Xv}) \label{eq:aor1-5} &\\
&\text{and} & R_2 &\le  f_2(\rho, K_{\Xv}) \label{eq:aor2}
\end{align}
This yields the rate region ${\Rc}_{\rm O}^{\rm MG}(\rho, K_{\Xv})$. Hence we
have the desired result.
\end{proof}

\bigskip

\begin{proof} ({\bf Lemma~\ref{lem:Ka}})
For a given $K_{\Xv}$, we first evaluate $L(K_{\Xv},0)$ and $L(K_{\Xv},1)$.
Since $\gamma_1(0)=1$ and the input covariance matrix $K_{\Xv}$ is positive
semidefinite, we obtain
\begin{align}
L(K_{\Xv},0)=(\hv-\rho_{\rm o}\gv)^{H}K_{\Xv}(\hv-\rho_{\rm o}\gv) \ge 0.
\end{align}

On the other hand, since $\gamma_1(1)=\lambda_1$ and $K_{\Xv}$ is Hermitian and
positive semidefinite, we have
\begin{align}
L(K_{\Xv},1)&=(\hv-\rho_{\rm o}\lambda_1\gv)^{H}(K_{\Xv}-P
\ev_1\ev_1^{H})(\hv-\rho_{\rm o}\lambda_1\gv) \notag\\
&\le \tr(K_{\Xv})|\hv-\rho_{\rm o}\lambda_1\gv|^2 -P|(\hv-\rho_{\rm
o}\lambda_1\gv)^{H}\ev_1|^2\notag\\
&\le P|\hv-\rho_{\rm o}\lambda_1\gv|^2 -P|(\hv-\rho_{\rm
o}\lambda_1\gv)^{H}\ev_1|^2. \label{eq:exist0}
\end{align}
Based on the definition of $\rho_{\rm o}$ in (\ref{eq:def-rhoo}), we can
compute
\begin{align}
\hv- \rho_{\rm o} \lambda_1 \gv &=\hv- \lambda_1\frac{
\gv\gv^{H}\ev_1}{\hv^{H}\ev_1}\notag\\
&=\frac{(\hv\hv^{H}- \lambda_1
\gv\gv^{H})\ev_1}{\hv^{H}\ev_1} \notag\\
&=\frac{(\lambda_1-1)\ev_1}{P\hv^{H}\ev_1}\label{eq:exist1}
\end{align}
where the last step follows from the definitions $\lambda_1$ and $\ev_1$ in
(\ref{eq:eig-def1}). Moreover, since $\ev_1^{H}\ev_1=1$, (\ref{eq:exist1}) can
be rewritten as
\begin{align}
L(K_{\Xv},1)&\le P\left|\frac{(\lambda_1-1)\ev_1}{P\hv^{H}\ev_1}\right|^2
-P\left|\frac{\lambda_1-1}{P\hv^{H}\ev_1}\ev_1^{H}\ev_1\right|^2\notag\\
&=0
\end{align}

We note that $L(K_{\Xv},\alpha)$ is a continuous function on the interval
$\alpha\in[0,1]$ for a give $K_{\Xv}$. Since $L(K_{\Xv},0) \ge 0$ and
$L(K_{\Xv},1) \le 0$, there exists $\alpha\in [0,1]$ such that
$L(K_{\Xv},\alpha)=0$.
\end{proof}

\bigskip

\begin{proof} ({\bf Lemma~\ref{lem:small}})
We note that $\lambda_1$ and $\ev_1$ are the largest generalized eigenvalue and
the corresponding normalized eigenvector of the pencil $[I+P\hv\hv^{H},
I+P\gv\gv^{H}]$. Based on the Rayleigh's quotient principle in
Theorem~\ref{thm:Q}, we obtain
\begin{align}
\max_{\cv}\frac{\cv^{H}(I+P\hv\hv^{H})\cv} {\cv^{H}(I+P\gv\gv^{H})\cv}
&=\frac{\ev_1^{H}(I+P\hv\hv^{H})\ev_1}
{\ev_1^{H}(I+P\gv\gv^{H})\ev_1}=\lambda_1.
\end{align}
Hence, we have
\begin{align}
\min_{\cv}\frac{(1+ \alpha P
|\gv^{H}\ev_1|^2)+\cv^{H}[(1-\alpha)P\gv\gv^{H}]\cv}{(1+ \alpha P
|\hv^{H}\ev_1|^2)+\cv^{H}[(1-\alpha)P\hv\hv^{H}]\cv} &=\frac
{\ev_1^{H}(I+P\gv\gv^{H})\ev_1} {\ev_1^{H}(I+P\hv\hv^{H})\ev_1}%\notag\\&
=\frac{1}{\lambda_1}.
\end{align}
By using the definition of $\gamma_1(\alpha)$ in (\ref{eq:def-gm-1}), we have
\begin{align}
\min_{\cv} \frac{\displaystyle \cv^{H} \left[I+\frac{(1-\alpha) P}{1+\alpha P
|\gv^{H}\ev_1|^2}\gv\gv^{H}\right]\cv}{\displaystyle
\cv^{H}\left[I+\frac{(1-\alpha) P}{1+\alpha P
|\hv^{H}\ev_1|^2}\hv\hv^{H}\right] \cv} & =\frac{\displaystyle \ev_1^{H}
\left[I+\frac{(1-\alpha) P}{1+\alpha P
|\gv^{H}\ev_1|^2}\gv\gv^{H}\right]\ev_1}{\displaystyle
\ev_1^{H}\left[I+\frac{(1-\alpha) P}{1+\alpha P
|\hv^{H}\ev_1|^2}\hv\hv^{H}\right] \ev_1} %\notag \\&
=\frac{\gamma_1(\alpha)}{\lambda_1}.
\end{align}
Now, the Rayleigh's quotient principle implies that
$\gamma_1(\alpha)/\lambda_1$ and $\ev_1(\alpha)$ are the smallest generalized
eigenvalue and the corresponding normalized eigenvector of the pencil
(\ref{eq:pencil2}).
\end{proof}

\bigskip

\begin{proof} ({\bf Lemma~\ref{lem:orth}})
We first show that $[\gv-\rho_{\rm o}^{*} \gamma_2(\alpha)\hv] \propto
\cv_2(\alpha).$ Since $\gamma_1(\alpha)$ is real and
\begin{align}
\rho_{\rm o}=\frac{1}{\gamma_1(\alpha)} \left[\frac{\hv^{H}\cv_2(\alpha)}
{\gv^{H}\cv_2(\alpha)}\right]^{*}, \label{eq:rhostar}
\end{align}
we have
\begin{align}
\gv-\rho_{\rm o}^{*} \gamma_2(\alpha)\hv&=\gv-
\frac{\gamma_2(\alpha)}{\gamma_1(\alpha)} \frac{\hv\hv^{H}\cv_2(\alpha)} {\gv^{H}\cv_2(\alpha)} \notag\\
&=\frac{[\gamma_1(\alpha)\gv\gv^{H}- \gamma_2(\alpha)
\hv\hv^{H}]\cv_2(\alpha)}{\gamma_1(\alpha)\gv^{H}\cv_2(\alpha)}.
\label{eq:gh-c2}
\end{align}
Note that (\ref{eq:g2c2}) implies that
\begin{align}
\left[\frac{(1-\alpha) P}{1+\alpha P
|\gv^{H}\ev_1|^2}\gv\gv^{H}-\gamma_2(\alpha) \frac{(1-\alpha) P}{1+ \alpha P
|\hv^{H}\ev_1|^2}\hv \hv^{H}\right] \cv_2(\alpha) =
[\gamma_2(\alpha)-1]\cv_2(\alpha).
\end{align}
Based on the definition of $\gamma_1(\alpha)$ in (\ref{eq:def-gm-1}), we obtain
\begin{align}
[\gamma_1(\alpha)\gv\gv^{H}- \gamma_2(\alpha)\hv \hv^{H}] \cv_2(\alpha)=
\frac{1+ \alpha P
|\hv^{H}\ev_1|^2}{(1-\alpha)P}[\gamma_2(\alpha)-1]\cv_2(\alpha).
\end{align}
Now, we can rewritten (\ref{eq:gh-c2}) as
\begin{align}
\gv-\rho_{\rm o}^{*} \gamma_2(\alpha)\hv&=\frac{(1+ \alpha P
|\hv^{H}\ev_1|^2)[\gamma_2(\alpha)-1]}{ \gamma_1(\alpha)\gv^{H}\cv_2(\alpha)
(1-\alpha)P } \cv_2(\alpha). \label{eq:gh-c22}
\end{align}
Hence, we obtain
\begin{align}
\frac{\gv-\rho_{\rm o}^{*} \gamma_2(\alpha)\hv}{|\gv-\rho_{\rm o}^{*}
\gamma_2(\alpha)\hv|^2}=\cv_2(\alpha). \label{eq:c2pro}
\end{align}

Next we prove that $[\hv- \rho_{\rm o}\gamma_1(\alpha)\gv]^{H}\cv_2(\alpha)=0.$
The definitions of $\lambda_1$ and $\ev_1$ in (\ref{eq:eig-def1}) implies that
\begin{align}
\hv\hv^{H}\ev_1= \left[\frac{\lambda_1-1}{P} I+\lambda_1\gv\gv^{H}\right]\ev_1.
\label{eq:pvot}
\end{align}
Substituting (\ref{eq:pvot}) into (\ref{eq:apx-m13}), we obtain
\begin{align}
\hv- \rho_{\rm o}
\gamma_1(\alpha)\gv&=\frac{1}{\hv^{H}\ev_1}\left[\frac{\lambda_1-1}{P}
I+\lambda_1\gv\gv^{H}-
\gamma_1(\alpha) \gv\gv^{H}\right]\ev_1 \notag\\
&=\frac{\lambda_1-1}{P\hv^{H}\ev_1}\left[I+\frac{\lambda_1-
\gamma_1(\alpha)}{\lambda_1-1} P \gv\gv^{H}\right]\ev_1. \label{eq:pvot1}
\end{align}
Based on the definition of $\gamma_1(\alpha)$ in (\ref{eq:def-gm-1}), we obtain
\begin{align}
\lambda_1- \gamma_1(\alpha)&=\lambda_1-\frac{1+\alpha P|\hv^{H}\ev_1|^2}{1+
\alpha P |\gv^{H}\ev_1|^2} \notag\\
&=\frac{(\lambda_1-1) +\alpha P(\lambda_1 |\gv^{H}\ev_1|^2 -|\hv^{H}\ev_1|^2)}{1+ \alpha P |\gv^{H}\ev_1|^2}\notag\\
&=(\lambda_1-1)\frac{1-\alpha}{1+ \alpha P |\gv^{H}\ev_1|^2} \label{eq:pvot2}
\end{align}
where the last step follows from (\ref{eq:hhgg}). Substituting (\ref{eq:pvot1})
into (\ref{eq:pvot2}), we have
\begin{align}
\hv- \rho_{\rm o} \gamma_1(\alpha)\gv
&=\frac{\lambda_1-1}{P\hv^{H}\ev_1}\left[I+\frac{(1-\alpha)P}{1+ \alpha P
|\gv^{H}\ev_1|^2} \gv\gv^{H}\right]\ev_1. \label{eq:pvot3}
\end{align}
Now (\ref{eq:e1ooc2}) and (\ref{eq:pvot3}) imply that
\begin{align}
[\hv- \rho_{\rm o} \gamma_1(\alpha)\gv]^{H}\cv_2(\alpha)=0.
\end{align}
Combining with (\ref{eq:c2pro}), we have the desired result.
\end{proof}

\bigskip

\begin{proof} ({\bf Equation (\ref{eq:fm6})})
First, we consider $\zeta(\alpha)$ defined in (\ref{eq:zta}). Based on
(\ref{eq:gh-c2}) and (\ref{eq:gh-c22}), $\zeta(\alpha)$ can be rewritten as
\begin{align}
\zeta(\alpha)&=[\gv-\rho_{\rm o}^{*} \gamma_2(\alpha)\hv]^{H}[\gv-\rho_{\rm
o}^{*} \gamma_2(\alpha)\hv] \notag\\
&=\frac{(1+ \alpha P |\hv^{H}\ev_1|^2)[\gamma_2(\alpha)-1]}{(1-\alpha)P}
\left\{\frac{\cv_2^{H}(\alpha) [\gamma_1(\alpha)\gv\gv^{H}-
\gamma_2(\alpha)\hv\hv^{H}]\cv_2(\alpha)}
{|\gamma_1(\alpha)\gv^{H}\cv_2(\alpha)|^2}\right\}  \notag\\
%&=\left[\frac{1+ \alpha P |\hv^{H}\ev_1|^2}{(1-\alpha)P}\right]
%\frac{[\gamma_2(\alpha)-1] [\gamma_1(\alpha)|\gv^{H}\cv_2(\alpha)|^2-
%\gamma_2(\alpha)|\hv^{H}\cv_2(\alpha)|^2]}
%{|\gamma_1(\alpha)\gv^{H}\cv_2(\alpha)|^2}  \notag\\
&=\left[\frac{1+ \alpha P |\hv^{H}\ev_1|^2}{(1-\alpha)P}\right]
[\gamma_2(\alpha)-1] \left[\frac{1}{\gamma_1(\alpha)}-
\gamma_2(\alpha)|\rho_{\rm o}|^2\right].\label{eq:fm2}
\end{align}
Now, we consider
\begin{align}
(1-\alpha)P \zeta(\alpha)&=[\gamma_2(\alpha)-1] [1+\alpha P |\gv^{H}\ev_1|^2-
\gamma_2(\alpha)|\rho_{\rm o}|^2- \alpha P |\hv^{H}\ev_1|^2
\gamma_2(\alpha)|\rho_{\rm o}|^2]\notag\\
&=[\gamma_2(\alpha)-1] \{1- \gamma_2(\alpha)|\rho_{\rm o}|^2+\alpha P
|\gv^{H}\ev_1|^2 [1- \gamma_2(\alpha)]\} \label{eq:fm5}
\end{align}
where the last step of (\ref{eq:fm5}) follows from $\rho_{\rm
o}=(\gv^{H}\ev_1)/(\hv^{H}\ev_1)$.

%Note that $\cv_2(\alpha)$ and $\ev_1$ are eigenvectors of the pencil
%\begin{align}
%\left(I+\frac{(1-\alpha) P}{1+\alpha P |\gv^{H}\ev_1|^2}\gv\gv^{H}, \;
%I+\frac{(1-\alpha) P}{1+ \alpha P |\hv^{H}\ev_1|^2}\hv \hv^{H}\right)
%\end{align}
%corresponding to eigenvalues $\gamma_2(\alpha)$ and
%$\gamma_1(\alpha)/\lambda_1$, respectively. The property of generalized
%eigenvalues (see Appendix~\ref{app:geig}) yields
%\begin{align}
%\ev_1^{H}\left[I+\frac{(1-\alpha) P}{1+\alpha P
%|\gv^{H}\ev_1|^2}\gv\gv^{H}\right]\cv_2(\alpha) &=0.
%\end{align}
%This implies
%\begin{align}
%\ev_1^{H}\cv_2(\alpha) &= -\frac{(1-\alpha) P}{1+\alpha P
%|\gv^{H}\ev_1|^2}\ev_1^{H}\gv\gv^{H}\cv_2(\alpha).
%\end{align}

Next, we consider $\eta(\alpha)$ defined in (\ref{eq:eta}). Note that
(\ref{eq:e1c2}) implies that
\begin{align}
|\ev_1^{H}\cv_2(\alpha)|^2 &= \left[\frac{(1-\alpha) P}{1+\alpha P
|\gv^{H}\ev_1|^2}\right]^2 |\gv^{H}\ev_1|^2|\gv^{H}\cv_2(\alpha)|^2.
\label{eq:2evcv}
\end{align}
Combining  (\ref{eq:rhostar}), (\ref{eq:fm2}) and (\ref{eq:2evcv}), we obtain
\begin{align}
\eta(\alpha)&= \zeta(\alpha)|\ev_1^{H}\cv_2(\alpha)|^2 \notag\\
&= \frac{(1+ \alpha P |\hv^{H}\ev_1|^2)(\gamma_2(\alpha)-1)}{(1-\alpha)P}
\left[\frac{1}{\gamma_1(\alpha)}-
\frac{\gamma_2(\alpha)}{\gamma_1^2(\alpha)}\left|\frac{\hv^{H}\cv_2}{\gv^{H}\cv_2}\right|^2
\right]
|\ev_1^{H}\cv_2(\alpha)|^2 \notag\\
&=[\gamma_2(\alpha)-1]|\gv^{H}\ev_1|^2
\left[\frac{(1-\alpha)P|\gv^{H}\cv_2(\alpha)|^2}{1+ \alpha P |\gv^{H}\ev_1|^2}-
\gamma_2(\alpha)\frac{(1-\alpha)P|\hv^{H}\cv_2(\alpha)|^2}{1+ \alpha P
|\hv^{H}\ev_1|^2}\right] \notag\\
&=[\gamma_2(\alpha)-1]^2|\gv^{H}\ev_1|^2 \label{eq:fm4}
\end{align}
where the last step of (\ref{eq:fm4}) follows from the definition of
$\gamma_2(\alpha)$ in (\ref{eq:g2c2}). Combining (\ref{eq:fm5}) and
(\ref{eq:fm4}), we obtain the desired result.
\end{proof}

\bibliographystyle{IEEEtran}
\bibliography{secrecy}

\begin{thebibliography}{10}
\providecommand{\url}[1]{#1}
\csname url@samestyle\endcsname
\providecommand{\newblock}{\relax}
\providecommand{\bibinfo}[2]{#2}
\providecommand{\BIBentrySTDinterwordspacing}{\spaceskip=0pt\relax}
\providecommand{\BIBentryALTinterwordstretchfactor}{4}
\providecommand{\BIBentryALTinterwordspacing}{\spaceskip=\fontdimen2\font plus
\BIBentryALTinterwordstretchfactor\fontdimen3\font minus
  \fontdimen4\font\relax}
\providecommand{\BIBforeignlanguage}[2]{{%
\expandafter\ifx\csname l@#1\endcsname\relax
\typeout{** WARNING: IEEEtran.bst: No hyphenation pattern has been}%
\typeout{** loaded for the language `#1'. Using the pattern for}%
\typeout{** the default language instead.}%
\else
\language=\csname l@#1\endcsname
\fi
#2}}
\providecommand{\BIBdecl}{\relax}
\BIBdecl

\bibitem{Wyner:BSTJ:75}
A.~D. Wyner, ``The wire-tap channel,'' \emph{Bell Syst. Tech. J.}, vol.~54,
  no.~8, pp. 1355--138, Oct. 1975.

\bibitem{Csiszar:IT:78}
I.~Csisz{\'{a}}r and J.~K{\"{o}}rner, ``Broadcast channels with confidential
  messages,'' \emph{IEEE Trans. Inf. Theory}, vol.~24, no.~3, pp. 339--348, May
  1978.

\bibitem{Oohama:ITW:01}
Y.~Oohama, ``Coding for relay channels with confidential messages,'' in
  \emph{Proc. IEEE Information Theory Workshop}, Cairns, Australia, Sep. 2001,
  pp. 87--89.

\bibitem{Csiszar:IT:04}
I.~Csisz{\'{a}}r and P.~Narayan, ``Secrecy capacities for multiple terminal,''
  \emph{IEEE Trans. Inf. Theory}, vol.~50, no.~12, pp. 3047--3061, Dec 2004.

\bibitem{Tekin:ISIT:06}
E.~Tekin and A.~Yener, ``The {G}aussian multiple access wire-tap channel with
  collective secrecy constraints,'' in \emph{Proc. IEEE Int. Symp. Information
  Theory (ISIT)}, Seattle, USA, Jul. 2006.

\bibitem{Tekin:ITA:07}
------, ``The multiple access wire-tap channel: Wireless secrecy and
  cooperative jamming,'' in \emph{Proc. Information Theory and Application
  Workshop, ITA}, San Diego, CA, Jan. 2007.

\bibitem{Liang06it}
\BIBentryALTinterwordspacing
Y.~Liang and H.~{Vincent~Poor}, ``Generalized multiple access channels with
  confidential messages,'' \emph{IEEE Trans. Inf. Theory}, submitted (under
  revision), April 2006. [Online]. Available:
  \url{http://arxiv.org/PS$\_$cache/cs/pdf/0605/0605014.pdf}
\BIBentrySTDinterwordspacing

\bibitem{Liu:ISIT:06}
R.~Liu, I.~Maric, R.~D. Yates, and P.~Spasojevic, ``The discrete memoryless
  multiple access channel with confidential messages,'' in \emph{Proc. IEEE
  Int. Symp. Information Theory (ISIT)}, Jul. 2006, pp. 957 -- 961.

\bibitem{Liu:Allerton:06}
R.~Liu, I.~Maric, P.~Spasojevic, and R.~Yates, ``Discrete memoryless
  interference and broadcast channels with confidential messages,'' in
  \emph{Proc. Allerton Conference on Communication, Control, and Computing},
  Sep. 2006.

\bibitem{Lai:IT:06}
L.~Lai and H.~{El Gamal}, ``The relay-eavesdropper channel: Cooperation for
  secrecy,'' \emph{IEEE Trans. Inf. Theory}, submitted, Dec. 2006.

\bibitem{Elza:CISS:07}
M.~Yuksel and E.~Erkip., ``The relay channel with a wiretapper,'' in
  \emph{Proc. Forty-First Annual Conference on Information Sciences and Systems
  (CISS)}, Baltimore, MD, USA, Mar. 2007.

\bibitem{Barros:ISIT:06}
J.~Barros and M.~Rodrigues, ``Secrecy capacity of wireless channels,'' in
  \emph{Proc. IEEE Int. Symp. Information Theory (ISIT)}, Seattle, USA, Jul.
  2006.

\bibitem{Li:Allerton:06}
Z.~Li, R.~Yates, and W.~Trappe, ``Secrecy capacity of indepedent parallel
  channels,'' in \emph{Proc. Allerton Conference on Commun., Contr.,
  Computing}, Monticello, IL, USA, Sep. 2006.

\bibitem{Liang06novit}
\BIBentryALTinterwordspacing
Y.~Liang, H.~{Vincent~Poor}, and {S. Shamai~(Shitz)}, ``Secure communication
  over fading channels,'' \emph{IEEE Trans. Inf. Theory}, submitted, Nov. 2006.
  [Online]. Available:
  \url{http://arxiv.org/PS$\_$cache/cs/pdf/0701/0701024.pdf}
\BIBentrySTDinterwordspacing

\bibitem{Gopala:ISIT:07}
P.~Gopala, L.~Lai, and H.~{El~Gamal}, ``On the secrecy capacity of fading
  channels,'' in \emph{Proc. IEEE Int. Symp. Information Theory (ISIT)}, Nice,
  France, June 24-29, 2007.

\bibitem{Li:CISS:07}
Z.~Li, W.~Trappe, and R.~Yates, ``Secret communication via multi-antenna
  transmission,'' in \emph{Proc. Forty-First Annual Conference on Information
  Sciences and Systems (CISS)}, Baltimore, MD, USA, Mar. 2007.

\bibitem{Liu:WITS:07}
R.~Liu and H.~{Vincent Poor}, ``Multiple antenna secure broadcast over wireless
  networks,'' in \emph{Proc. First International Workshop on Information Theory
  for Sensor Networks}, Santa Fe, NM, June 18-20, 2007, pp. 125--139.

\bibitem{Wornell:ISIT:07}
A.~Khisti, G.~Wornell, A.~Wiesel, and Y.~Eldar, ``On the {G}aussian {MIMO}
  wiretap channel,'' in \emph{Proc. IEEE Int. Symp. Information Theory (ISIT)},
  Nice, France, June 24-29, 2007.

\bibitem{Wornell-IT-2007}
A.~Khisti and G.~Wornell, ``Secure transmission with multiple antennas: The
  {MISOME} wiretap channel,'' \emph{IEEE Trans. Inf. Theory}, submitted, August
  2007.

\bibitem{Ulukus:ISIT:07}
S.~Shafiee and S.~Ulukus, ``Achievable rates in gaussian {MISO} channels with
  secrecy constraints,'' in \emph{Proc. IEEE Int. Symp. Information Theory
  (ISIT)}, Nice, France, June 24-29, 2007.

\bibitem{shafiee-IT-2007}
S.~Shafiee, N.~Liu, and S.~Ulukus, ``Towards the secrecy capacity of the
  {G}aussian {MIMO} wire-tap channel: The 2-2-1 channel,'' \emph{IEEE Trans.
  Inf. Theory}, submitted, September 2007.

\bibitem{Cover}
T.~Cover and J.~Thomas, \emph{Elements of Information Theory}.\hskip 1em plus
  0.5em minus 0.4em\relax New York: John Wiley Sons, Inc., 1991.

\bibitem{Cheong:IT:78}
S.~K. Leung-Yan-Cheong and M.~E. Hellman, ``The {G}aussian wire-tap channel,''
  \emph{IEEE Trans. Inf. Theory}, vol.~24, no.~4, pp. 51--456, Jul. 1978.

\bibitem{Liu07it}
\BIBentryALTinterwordspacing
R.~Liu, I.~Maric, P.~Spasojevic, and R.~Yates, ``Discrete memoryless
  interference and broadcast channels with confidential messages: Secrecy rate
  regions,'' \emph{IEEE Trans. Inf. Theory}, submitted, Feb 2007. [Online].
  Available: \url{http://arxiv.org/PS$\_$cache/cs/pdf/0702/0702099.pdf}
\BIBentrySTDinterwordspacing

\bibitem{Maurer:EUROCRYPT:00}
U.~Maurer and S.~Wolf, ``Information-theoretic key agreement: From weak to
  strong secrecy for free,'' in \emph{Proc. EUROCRYPT, Lecture Notes in
  Computer Science}, vol. 1807, 2000, pp. 351--368.

\bibitem{Slepian:IT:73}
D.~Slepian and J.~K. Wolf, ``Noiseless coding of correlated information
  sources,'' \emph{IEEE Trans. Inf. Theory}, vol.~19, no.~4, pp. 471--480, Jul.
  1973.

\bibitem{marton:IT:77}
K.~Marton, ``A coding theorem for the discrete memoryless broadcast channel,''
  \emph{IEEE Trans. Inf. Theory}, vol.~25, pp. 306--311, May 1979.

\bibitem{Caire:IT:03}
G.~Caire and {S.~Shamai~(Shitz)}, ``On the achievable throughput of a
  multiantenna {G}aussian broadcast channel,'' \emph{IEEE Trans. Inf. Theory},
  vol.~49, no.~7, pp. 1691--1706, Jul. 2003.

\bibitem{Yu:IT:04}
{W.~Yu} and {J.~M.~Cioffi}, ``Sum capacity of {G}aussian vector broadcast
  channels,'' \emph{IEEE Trans. Inf. Theory}, vol.~50, pp. 1875--1892, Sep.
  2004.

\bibitem{Gilbert}
G.~Strang, \emph{Linear Algebra and Its Applications}.\hskip 1em plus 0.5em
  minus 0.4em\relax Wellesley, MA: Wellesley-Cambridge Press, 1998.

\end{thebibliography}

\end{document}